\documentclass[envcountsame]{llncs}

\usepackage{multirow}
\usepackage{makecell}
\usepackage{xspace}
\usepackage{float}
\usepackage{bm}
\usepackage{slashbox}
\usepackage{amsfonts}
\usepackage{amsmath}
\usepackage{bbm}
\usepackage{graphicx}
\usepackage{enumitem}
\usepackage{subfigure}

\usepackage{xfrac}
\usepackage{color}
\newtheorem{algorithm}{Algorithm}
\newtheorem{clm}[theorem]{Claim}{\bfseries}{\itshape}

\newcommand{\commentout}[1]{}
\newcommand{\commentin}[1]{#1}
\newcommand{\mysubsub}[1]{\vskip 2pt \noindent {\bf #1.}}
\newcommand{\E}{\mathbb{E}}

\newcommand{\mdp}{M}
\newcommand{\rh}{_{\rho}}
\newcommand{\starr}{\ensuremath{*}}

\newcommand{\sa}{\ensuremath{a}\xspace}
\newcommand{\sh}{\ensuremath{h}\xspace}

\newcommand{\lt}{Low\text{-}Triangle}

\title{Optimal Selfish Mining Strategies in Bitcoin}

\author{Ayelet Sapirshtein \inst{1} \and Yonatan Sompolinsky\inst{1} \and Aviv Zohar\inst{1,2}}

\institute{School of Engineering and Computer Science,\newline The Hebrew University of Jerusalem, Israel\and Microsoft Research, Herzliya, Israel \\ \email{\{ayeletsa,yoni\_sompo,avivz\}@cs.huji.ac.il}}

\newcommand{\ad}{\ensuremath{adopt}\xspace}
\newcommand{\ma}{\ensuremath{match}\xspace}
\newcommand{\ov}{\ensuremath{override}\xspace}
\newcommand{\wa}{\ensuremath{wait}\xspace}

\newcommand{\rev}{\ensuremath{REV}}

\newcommand{\sm}{SM1\xspace}
\newcommand{\hon}{\text{honest mining}\xspace}

\begin{document}
\maketitle

\begin{abstract}
Bitcoin is a decentralized crypto-currency, and an accompanying protocol, created in 2008.
Bitcoin nodes continuously generate and propagate blocks---collections of newly approved transactions that are added to Bitcoin's ledger. Block creation requires nodes to invest computational resources, but also carries a reward in the form of bitcoins that are paid to the creator.
While the protocol requires nodes to quickly distribute newly created blocks,
strong nodes can in fact gain higher payoffs by withholding blocks they create and selectively postponing their publication. The existence of such \emph{selfish mining} attacks was first reported by Eyal and Sirer~\cite{ES}, who have demonstrated a specific deviation from the standard protocol (a strategy that we name \sm).

In this paper we extend the underlying model for selfish mining attacks, and provide an algorithm to find $\epsilon$-optimal policies for attackers within the model, as well as tight upper bounds on the revenue of optimal policies.
As a consequence, we are able to provide lower bounds on the computational power an attacker needs in order to benefit from selfish mining. We find that the \emph{profit threshold} -- the minimal fraction of resources required for a profitable attack -- is strictly lower than the one induced by the \sm scheme. Indeed, the policies given by our algorithm dominate \sm, by better regulating attack-withdrawals.

Our algorithm can also be used to evaluate protocol modifications that aim to reduce the profitability of selfish mining. We demonstrate this with regard to a suggested countermeasure by Eyal and Sirer, and show that it is slightly less effective than previously conjectured.  Next, we gain insight into selfish mining in the presence of communication delays, and show that, under a model that accounts for delays, the profit threshold vanishes, and even small attackers have incentive to occasionally deviate from the protocol. We conclude with observations regarding the combined power of selfish mining and double spending attacks.
\end{abstract}

\section{Introduction}
In a recent paper, Eyal and Sirer~\cite{ES} have highlighted a flaw in the incentive scheme in Bitcoin. Given that most of the network follows the ``standard'' Bitcoin protocol, a single node (or a pool) which possesses enough computational resources or is extremely well connected to the rest of the network can increase its expected rewards by deviating from the protocol. While the standard Bitcoin protocol requires nodes to immediately publish any block that they find to the rest of the network, Eyal and Sirer have shown that participants can selfishly increase their revenue by selectively withholding blocks.
Their strategy, which we denote \sm, thus shows that Bitcoin as currently formulated is not incentive compatible.

On the positive side, \sm (under the model of Eyal and Sirer) becomes profitable only when employed by nodes that posses a large enough share of the computational resources, and are sufficiently well connected to the rest of the network.\footnote{This can partly explain why selfish mining attacks have not been observed in the Bitcoin network thus far.}
It is important to note, however, that \sm is not the optimal best-response to honest behaviour, and situations in which \sm is not profitable may yet have other strategies that are better than strict adherence to the protocol. Our goal in this paper is to better understand the conditions under which Bitcoin is resilient to selfish mining attacks. To this end, we must consider other possible deviations from the protocol, and to establish bounds on their profitability.

The role of incentives in Bitcoin should not be underestimated:
Bitcoin transactions are confirmed in batches, called \emph{blocks} whose creation requires generating the solution to computationally expensive proof-of-work ``puzzles''. The security of Bitcoin against the reversal of payments (so-called double spending attacks) relies on having more computational power in the hands of honest nodes. Block creation (which is also known as \emph{mining}), is rewarded in bitcoins that are given to the block's creator. These rewards incentivize more honest participants to invest additional computational resources in mining, and thus support the security of Bitcoin.

When all miners follow the Bitcoin protocol, a single miner's share of the payoffs is equal to the fraction of computational power that it controls (out of the computational resources of the entire network).
However, Selfish mining schemes allow a strong attacker to increase its revenue at the expense of other nodes. This is done by exploiting the conflict-resolution rule of the protocol, according to which only one chain of blocks can be considered valid, and only blocks on the valid chain receive rewards; the attacker creates a deliberate fork, and (sometimes) manages to force the honest network to abandon and discard some of its blocks.

The consequences of selfish mining attacks are potentially destructive to the Bitcoin system. A successful attacker becomes more profitable than honest nodes, and is able to grow steadily.\footnote{Growth is achieved either by buying more hardware, in the case of a single attacker, or by attracting more miners, in the case of a pool.\label{ftnt::growth}} It may thus eventually drive other nodes out of the system. Profits from selfish mining increase as more computational power is held by the attacker, making its attack increasingly effective, until it eventually holds over 50\% of the computational resources in the network. At this point, the attacker is able to collect all block rewards, to mount successful double spending attacks at will, and to block any transaction from being processed (this is known as the 50\% attack).


We summarize the contributions of this paper as follows:
\begin{enumerate}
\item We provide an efficient algorithm that computes an $\epsilon$-optimal selfish mining policy for any $\epsilon>0$, and for any parametrization of the model in~\cite{ES} (i.e., one that maximizes the revenue of the attacker up to an error of $\epsilon$, given that all other nodes are following the standard Bitcoin protocol). We prove the correctness of our algorithm and analyze its error bound. We further verify all strategies generated by the algorithm in a selfish mining simulator that we have designed to this end.
\item Using our algorithm we show that, indeed, there are selfish mining strategies that earn more money and are profitable for smaller miners compared to \sm. The gains are relatively small (see Fig.~\ref{fig::res_rhos} below). This can be seen as a positive result, lower bounding the amount of resources needed for a profitable attacker.
\item Our technique allows us to evaluate different protocol modifications that were suggested as countermeasures for selfish mining. We do so for the solution suggested by Eyal and Sirer, in which miners that face two chains of equal weight choose the one to extend uniformly at random. We show that this modification unexpectedly enhances the power of medium-sized attackers, while limiting strong ones, and that unlike previously conjectured, attackers with less than 25\% of the computational resources can still gain from selfish mining.
\item We show that in a model that accounts for the delay of block propagation in the network, the threshold vanishes: there is always a successful selfish mining strategy that earns more than honest mining, regardless of the size of the attacker.
\item We discuss the interaction between selfish mining attacks and double spending attacks. We demonstrate how any attacker for which selfish mining is profitable can execute double spending attacks bearing no costs. This sheds light on the security analysis of Satoshi Nakamoto~\cite{SATOSHI}, and specifically, on the reason that it cannot be used to show high attack costs, and must instead only bound the probability of a successful attack.
\end{enumerate}


Below, we depict the results of our analysis, namely, the revenue achieved by optimal policies compared to that of \sm as well as the profit threshold of the protocol. In the following, $\alpha$ stands for the attacker's relative hashrate, and $\gamma$ is a parameter representing the communication capabilities of the attacker: the fraction of nodes to which it manages to send blocks first in case of a block race (see Section~\ref{sec::model} for more details).
Figure~\ref{fig::res_rhos} depicts the revenue of an attacker under three strategies: Honest mining, which adheres to the Bitcoin protocol, \sm, and the optimal policies obtained by our algorithm. The three graphs correspond to $\gamma=0,0.5,1$. We additionally illustrate the curve of $\alpha/(1-\alpha)$, which is an upper bound on the attacker's revenue, achievable only when $\gamma=1$ (see Section~\ref{sec::upperbound}).
Figure~\ref{fig::thresholds} depicts the profit threshold for each $\gamma$: If the attacker's $\alpha$ is below the threshold then Honest mining is the most profitable strategy. For comparison, we depict the thresholds induced by \sm as well.

\begin{figure}[!h]
\centering
\subfigure[$\gamma=0$]{
\includegraphics[scale=0.417]{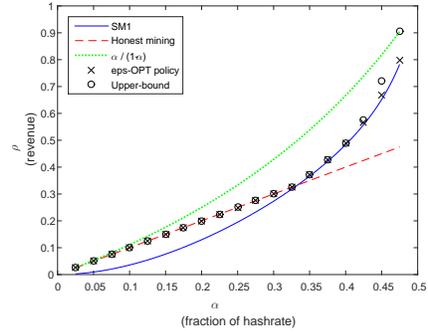}
}
\subfigure[$\gamma=0.5$]{
\includegraphics[scale=0.417]{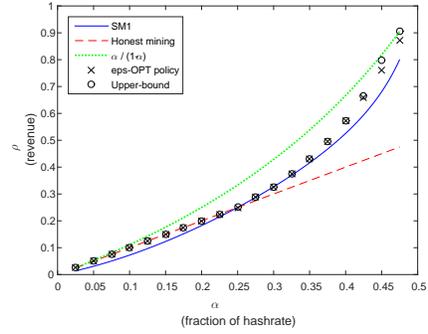}
}
\subfigure[$\gamma=1$]{
\includegraphics[scale=0.417]{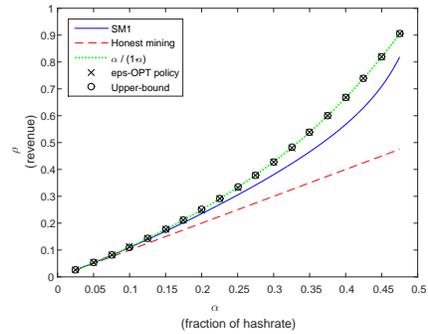}
}
\caption{The $\epsilon$-optimal revenue and the computed upper bound, as a function of the attacker's hashrate $\alpha$, compared to \sm, honest mining, and to the hypothetical bound provided in Section~\ref{sec::upperbound}. The graphs differ in the attacker's communication capability, $\gamma$, valued $0$, $0.5$, and $1$. The gains of the $\epsilon$-optimal policies are very close to the computed upper bound, except when $\alpha$ is close to $0.5$, in case which the truncation-imposed loss is apparent.
See also Table~\ref{table::revenue_gamma_0}.
}
\label{fig::res_rhos}
\end{figure}
\begin{figure}[!h]
\centering
\includegraphics[scale=0.6]{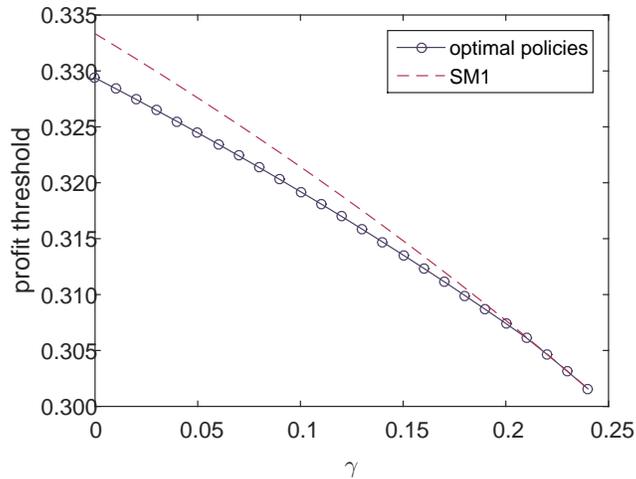}
\caption{The profit thresholds induced by optimal policies, and by \sm, as a function of $\gamma$. Thresholds at higher $\gamma$ values match that of \sm (but still, optimal strategies for these values earn more than \sm, once above the threshold).}
\label{fig::thresholds}
\end{figure}

The remainder of the paper is structured as follows: We begin by presenting our model, based principally on Eyal and Sirer's~\cite{ES} (Section~\ref{sec::model}). Section~\ref{sec::upperbound} shows a theoretical bound on the attacker's revenue. In Section~\ref{sec::alg} we describe our algorithm to find optimal policies and values. In Section~\ref{sec::additional_results} we discuss more results, e.g., the optimal policies. Section~\ref{sec::delays} analyzes selfish mining in networks with delays. Section~\ref{sec::doublespend} discusses the interaction between selfish mining and double spending. We conclude with discussing related work (Section~\ref{sec::related}).

\section{Model}\label{sec::model}
We follow and extend the model of~\cite{ES}, to explicitly consider all
 actions available to the attacker at any given point in time.

We assume that the attacker controls a fraction $\alpha$ of the computational power in the network, and that the honest network thus has a $(1-\alpha)$ fraction. Communication of newly created blocks is modeled to be much faster than block creation, so no blocks are generated while others are being transmitted.\footnote{This is justified by Bitcoin's 10 minute block creation interval which is far greater than the propagation time of blocks in the network. This assumption is later removed when we consider networks with delay.}

Blocks are created in the network according to a Poisson process with rate $\lambda$. Every new block is generated by the attacker with probability $\alpha$, or by the honest network with probability $(1-\alpha)$. The honest network follows the Bitcoin protocol, and always builds its newest block on top of the longest known chain. Once an honest node adopts a block, it will discard it only if a strictly longer competing chain exists. Ties are thus handled by each node according to the order of arrival of blocks. Honest nodes immediately broadcast blocks that they create.

Blocks generally form a tree structure, as each block references a single predecessor (with the exception of the first block that is called the genesis block). Since the honest nodes adopt the longest chain, blocks generate rewards for their creator only if they are eventually part of the longest chain in the block tree (all blocks can be considered revealed eventually).

To model the communication capabilities of the attacker, we assume that whenever it learns that a block has been released by the network, it is able to transmit an alternative block which will arrive \emph{first} at nodes that possess a fraction $\gamma$ of the computational power of the honest network (the attacker must have prepared this block in advance in order to be able to deliver it quickly enough). Thus, if the network is currently propagating a block of height $h$, and the attacker has a competing block of the same height, it is able to get $\gamma\cdot (1-\alpha)$ of the computational power (owned by honest nodes) to adopt this block.

The attacker does not necessarily follow the Bitcoin protocol. Rather, at any given time $t$, it may choose to invest computational power in creating blocks that extend any existing block in history, and may withhold blocks it has created for any amount of time. A general selfish mining strategy dictates, therefore, two key behaviours: which block the attacker attempts to extend at any time $t$, and which blocks are released at any given time.
However, given that all block creation events are driven by memoryless processes and that broadcast is modeled as instantaneous, any rational decision made by the attacker may only change upon the creation of a new block. The mere passage of time without block creation does not otherwise alter the expected gains from future outcomes.\footnote{See Section~\ref{sec::delays} for the implication of delayed broadcasting.} Accordingly, we model the entire decision problem faced by an attacker using a discrete-time process in which each time step corresponds to the creation of a block. The attacker is thus asked to decide on a course of action immediately after the creation of each block, and this action is pursued until the next event occurs.

Instead of directly modeling the primitive actions of block extension and publication on general block trees, we can limit our focus to ``reasonable'' strategies where the attacker maintains a \emph{single} secret branch of blocks that diverged from the network's chain at some point. (We show that this limitation is warranted and that this limited strategy space still generates optimal attacks in Appendix~\ref{appendix::model}). Blocks before that point are agreed upon by all participants.
Accordingly, we must only keep track of blocks that are after the fork, and of the accumulated reward up to the fork.
We denote by $\sa$ the number of blocks that have been built by the attacker after the latest fork, and by $\sh$ the number of those built by honest nodes.

Formally, if all other participants are following the standard protocol, the attacker faces a single-player decision problem of the form $M:=\langle S,A,P,R\rangle$, where $S$ is the state space, $A$ the action space, $P$ a stochastic transition matrix that describes the probability of transitioning between states, and $R$ the reward matrix.
Though similar in structure, we do not regard $M$ as an MDP, since the objective function is nonlinear: The player aims to maximize its share of the accepted blocks, rather than the absolute number of its own accepted ones; its goal is to have a greater return-on-investment than its counterparts.\footnote{Another possible motivation for this is the re-targeting mechanism in Bitcoin. When the block creation rate in the network is constant, the adaptive re-targeting implies that the attacker will also increase its absolute payoff, in the long run.}

\mysubsub{Actions}
We begin with the description of the action space $A$, which will motivate the nontrivial construction of the state space.
\begin{itemize}
\item[$\bullet$] {\bf Adopt}. The action $\ad$ is always feasible, and represents the attacker's acceptance of the honest network's chain. The $\sa$ blocks in the attacker's current chain are discarded.
\item[$\bullet$] {\bf Override}. The action $\ov$ represents the publication of the attacker's blocks, and is feasible whenever $\sa>\sh$.
\item[$\bullet$] {\bf Match}. This action represents the case where the most recent block was built by the honest network, and the attacker now publishes a conflicting block of the same height. This action is not always feasible (the attacker must have a block prepared in advance to execute such a race). The state-space explicitly encodes the feasibility status of this action (see below).
\item[$\bullet$] {\bf Wait}. Lastly, the $\wa$ action, which is always feasible, implies that the attacker does not publish new blocks, but keeps working on its branch until a new block is built.
\end{itemize}

\mysubsub{State Space} The state space, denoted $S$, is defined by 3-tuples of the form $(\sa,\sh,fork)$. The first two entries represent the lengths of the attacker's chain and the honest network's chain, built after the latest fork (that is, above the most recent block accepted by all). The field $fork$ obtains three values, dubbed $irrelevant$, $relevant$ and $active$. State of the form $(\sa,\sh,relevant)$ means that the previous state was of the form $(\sa,\sh-1,\cdot)$; this implies that if $\sa\ge\sh$, \ma is feasible. Conversely, $(\sa,\sh,irrelevant)$ denotes the case where the previous state was $(\sa-1,\sh,\cdot)$, rendering \ma now ineffective, as all honest nodes received already the $\sh$'th block. The third label, $active$, represents the case where the honest network is already split, due to a previous \ma action; this information affects the transition to the next state, as described below.
\emph{We will refer to states as $(\sa,\sh)$ or $(\sa,\sh,\cdot)$, in contexts where the $fork$ label plays no effective role.}

\mysubsub{Transition and Reward Matrices}
In order to keep the time averaging of rewards in scale, every state transition corresponds to the creation of a new block. The initial state $X_0$ is $(1,0,irrelevant)$ w.p. $\alpha$ or $(0,1,irrelevant)$ w.p. $(1-\alpha)$.
Rewards are given as elements in $\mathbb{N}^2$, where the first entry represents blocks of the attacker that have been accepted by all parties, and the second one, similarly, for those of the honest network.

The transition matrix $P$ and reward matrix $R$ are succinctly described in Table~\ref{table::PandR}. Largely, an \ad action ``resets'' the game, hence the state following it has the same distribution as $X_0$; its immediate reward is $\sh$ in the coordinate corresponding to the honest network. An \ov reduces the attacker's secret chain by $\sh+1$ blocks, which it publishes, and which the honest network accepts. This bestows a reward of $\sh+1$ blocks to the attacker. The state following a \ma action depends on whether the next block is created by the attacker ($\alpha$), by honest nodes working on their branch of the chain ($(1-\gamma)\cdot(1-\alpha)$), or by an honest node which accepted the sub-chain that the attacker published ($\gamma\cdot(1-\alpha)$). In the latter case, the attacker has effectively overridden the honest network's previous chain, and is awarded $\sh$ accordingly.

\begin{table}[h]
\caption{A description of the transition and reward matrices $P$ and $R$ in the decision problem $M$. The third column contains the probability of transiting from the state specified in the left-most column, under the action specified therein, to the state on the second one. The corresponding two-dimensional reward (the reward of the attacker and that of the honest nodes) is specified on the right-most column.\label{table::PandR}}
\begin{tabular}{|c|c|c|c|}
\hline
\textbf{State $\times$ Action} & \textbf{State} & \textbf{Probability} & \textbf{Reward}\\
\hline
\multirow{2}{*}{$(\sa,\sh,\cdot),adopt$} & $(1,0,irrelevant)$ & $\alpha$ & \multirow{2}{*}{$(0,\sh)$}\\
\cline{2-3} & $(0,1,irrelevant)$ & $1-\alpha$ & \\
\hline
\multirow{2}{*}{$(\sa,\sh,\cdot),override^\dagger$} & $(\sa-\sh,0,irrelevant)$ & $\alpha$ & \multirow{2}{*}{$(\sh+1,0)$}\\
\cline{2-3} & $(\sa-\sh-1,1,relevant)$ & $1-\alpha$ & \\
\hline
\multirow{2}{*}{\makecell{$(\sa,\sh,irrelevant),wait$\ \\ $(\sa,\sh,relevant),wait$}} & $(\sa+1,\sh,irrelevant)$ & $\alpha$ & (0,0)\\
\cline{2-4} & $(\sa,\sh+1,relevant)$ & $1-\alpha$ & (0,0)\\
\hline
\multirow{3}{*}{ \makecell{$(\sa,\sh,active),wait$\ \\ $(\sa,\sh,relevant),match^\ddagger$}} & $(\sa+1,\sh,active)$ & $\alpha$ & (0,0)\\
\cline{2-4} & $(\sa-\sh,1,relevant)$  & $\gamma\cdot(1-\alpha)$ & $(\sh,0)$\\
\cline{2-4} & $(\sa,\sh+1,relevant)$ & $(1-\gamma)\cdot(1-\alpha)$ & (0,0)\\
\hline
\multicolumn{1}{l}{\multirow{2}{*}{\begin{scriptsize}$^\dagger$feasible only when $\sa>\sh$\end{scriptsize}}}\\
    \multicolumn{1}{l}{\begin{scriptsize}$^\ddagger$feasible only when $\sa\ge \sh$\end{scriptsize}}\\
\end{tabular}
\end{table}

\mysubsub{Objective Function} As explained in the introduction, the attacker aims to maximize its relative revenue, rather than its absolute one as usual in MDPs. Let $\pi$ be a policy of the player; we will write $\pi(\sa,\sh,fork)$ for the action that $\pi$ dictates be taken at state $(\sa,\sh,fork)$. 
Denote by $X_t^\pi$ the state visited by time $t$ under $\pi$, and let $r(x,y,\pi)=(r^1(x,y,\pi),r^2(x,y,\pi))$ be the immediate reward from transiting from state $x$ to state $y$, under the action dictated by $\pi$. 
$X^\pi_t$ will denote the $t$'th state that was visited. We will abbreviate $r_t(X^\pi_{t},X^\pi_{t+1},\pi)$ and write simply $r_t(\pi)$ or even $r_t$, when context is clear.
The objective function of the player is its \emph{relative} payoff, defined by
\begin{align}\label{eq::revenue_definition} &\rev:=\E \left[\liminf\limits_{T\rightarrow\infty} \frac{\sum_{t=1}^Tr^1_t(\pi)}{\sum_{t=1}^T\left(r^1_t(\pi)+r^2_t(\pi)\right)}\right].
\end{align}

We will specify the parameters of $\rev$ depending on the context (e.g., $\rev(\pi,\alpha,\gamma)$, $\rev(\pi)$, $\rev(\alpha)$), and will occasionally denote the value of $\rev$ by $\rho$. In addition, for full definiteness of \rev, we rule out pathological behaviours in which the attacker waits forever---formally, the expected time for the next non-null action of the attacker must be finite.
%

\mysubsub{Honset Mining and \sm}
We now define two policies of prime interest to this paper. Honest mining is the unique policy which adheres to the protocol at every state. It is defined by
\begin{equation}
\label{eq::HONESTdef}
\hon\left(\sa,\sh,\cdot\right) = \left\{ \begin{array}{ccc}
\ad \quad \quad & \sh>\sa\\
\ov \quad \quad & \sa>\sh
\end{array} \right\},
\end{equation}
and \wa otherwise.
Notice that under our model, $\rev(\hon,\alpha,\gamma)=\alpha$ for all $\gamma$.\footnote{Indeed, in networks without delay, \hon is equivalent to the policy
$\left\{ \begin{array}{ccc}
\ad \quad & \text{if}~(\sa,\sh)=(0,1) ~;~
\ov \quad & \text{if}~(\sa,\sh)=(1,0)
\end{array} \right\}$, as these are the only reachable states. Delays allow other states to be reached, and will be covered in Section~\ref{sec::delays}.}
Eyal and Sirer's selfish mining strategy, \sm, can be defined as
\begin{equation}
\label{eq::ESdef}
\sm\left(\sa,\sh,\cdot\right) := \left\{ \begin{array}{ccc}
\ad \quad \quad & \sh>\sa\\
\ma \quad \quad & \sh=\sa=1\\
\ov \quad \quad & \sh=\sa-1\ge1\\
\wa \quad \quad & \text{otherwise}
\end{array} \right\}.
\end{equation}

\mysubsub{Profit threshold}
Keeping the attacker's connectivity capabilities ($\gamma$) fixed, we are interested in the minimal $\alpha$ for which employing dishonest mining strategies becomes profitable.
We define the profit threshold by:
 \begin{equation}
\hat\alpha(\gamma):=\inf_{\alpha}\left\{\exists\pi\in A :\:\rev(\pi,\alpha,\gamma)>\rev(\hon,\alpha,\gamma)\right\}.\label{eq::threshold_def}
\end{equation}


\section{A Simple Upper Bound}\label{sec::upperbound}
The mechanism implied by the longest-chain rule leads to an immediate bound on the attacker's relative revenue. 
Intuitively, we observe that the attacker cannot do better than utilizing every block it creates to override one block of the honest network.
The implied bound is provided here merely for general insight---it is usually far from the actual maximal revenue.
\begin{proposition}\label{prop::upperbound}
For any $\pi$, $\rev(\pi,\alpha,\gamma)\le\frac{\alpha}{1-\alpha}$. Moreover, this bound is tight, and achieved when $\gamma=1$.
\end{proposition}
See Appendix~\ref{appendix::upperbound} for the proof.

\section{Solving for the Optimal Policy}\label{sec::alg}
Finding an optimal policy is not a trivial task, as the objective function~(\ref{eq::revenue_definition}) is nonlinear, and depends on the entire history of the game. To overcome this we introduce the following method. We assume first that the optimal value of the objective function is $\rho$, then construct an infinite un-discounted average reward MDP (with ``standard'' linear rewards), compute its optimal policy (using standard MDP solution techniques), and if the reward of this policy is zero then it is optimal also in the original decision problem $M$. We elaborate on this approach below.

\subsection{Method}
For any $\rho\in [0,1]$, define the transformation $w\rh:\mathbb{N}^2\rightarrow\mathbb{Z}$ by $w\rh(x,y):=(1-\rho)\cdot x-\rho\cdot y$. Define the MDP $\mdp\rh:=\langle S,A,P,w\rh(R)\rangle $; it shares the same state space, actions, and transition matrix as $\mdp$, while $\mdp$'s immediate rewards matrix is transformed according to $w\rh$. For any admissible policy $\pi$ denote by
$v^\pi_{\rho}$ the expected mean revenue under $\pi$, namely,
\begin{equation}
\label{eq::value_rho}
v^\pi_{\rho}=\E\left[\liminf\limits_{T\rightarrow\infty}\frac1T\sum_{t=1}^T w\rh(r_t(\pi))\right],
\end{equation}
and by
\begin{equation}
\label{eq::value_rho_max}
v^*_{\rho}=\max_{\pi\in A}\left\{v^\pi\rh\right\}
\end{equation}
the value of $\mdp\rh$.\footnote{The equivalence of this formalization of the value function and alternatives in which the order of expectation and limit is reversed is discussed in~\cite{BIERTH}.} Our solution method is based on the following proposition: 
\begin{proposition}\label{prop::method}
\begin{enumerate}
\item
If for some $\rho\in[0,1]$, $v^*_{\rho}=0$, then any policy $\pi^*$ obtaining this value (thus maximizing $v^{\pi}\rh$) also maximizes $\rev$, and $\rho=\rev(\pi^*\rh)$.
\item $v^*\rh$ is monotonically decreasing in $\rho$.
\end{enumerate}
\end{proposition}

Following these observations we can utilize the family $\mdp\rh$ to obtain an optimal policy: We perform a simple search for a $\rho$ such that the optimal solution of $\mdp\rh$ has a value of 0. Since $v^*\rh$ is monotonically decreasing, this search can be done efficiently, using binary search. In practice, our algorithm relies on a variation of Proposition~\ref{prop::method}, which will be proven formally in Appendix~\ref{appendix::correctness}.

Due to the fact that the search domain is continuous, practically, one would need to halt the search at a point that is sufficiently close to the actual value, but never exact. Moreover, in practice, MDP solvers can solve only finite state space MDPs, and even then only to a limited degree of accuracy.
Our algorithm copes with these computational limitations by using finite MDPs as bounds to the original problem, and by analyzing the potential error that is due to inexact solutions.

\subsection{Translation to Finite MDPs}
We now introduce two families of MDPs, closely related to the family $\mdp\rh$: Fix some $T\in\mathbb{N}$. We define an \emph{under-paying} MDP, $\mdp\rh^T$, which differs from $\mdp\rh$ only in states where $\max\left\{\sa,\sh\right\}=T$, in which it only allows only for the \ad action. We denote this modified action space by $A^T$.
Clearly, the player's value in $\mdp^T\rh$ lower bounds that in $\mdp\rh$, since in the latter the attacker might be able to do better by not adopting in the truncating states. Consequently, this MDP can only be used to upper bound the threshold (in a way described below).

To complete the picture we need to bound the optimal value from above, and we do so by constructing an \emph{over-paying} MDP, $N^T\rh$. This MDP shares the same constraint as $\mdp^T\rh$, yet it compensates the attacker in the states where $\max\left\{\sa,\sh\right\}=T$, by granting it a reward greater than what it could have gotten in the un-truncated process:
When $T=\sa\ge\sh$, the attacker is awarded
\begin{equation}
(1-\rho)\cdot\frac{\alpha\cdot(1-\alpha)}{\left(1-2\cdot\alpha\right)^2}+\frac12\cdot\left(\frac{\sa-\sh}{1-2\cdot\alpha}+\sa+\sh\right).
\end{equation}
On the other hand, when $T=\sh\ge\sa$, it is awarded
\begin{align*}
&\left(1-\left(\frac{\alpha}{1-\alpha}\right)^{\sh-\sa}\right)\cdot\left(-\rho\cdot\sh\right)+\left(\frac{\alpha}{1-\alpha}\right)^{\sh-\sa}\cdot(1-\rho)\cdot\left(\frac{\alpha\cdot(1-\alpha)}{\left(1-2\cdot\alpha\right)^2}+\frac{\sh-\sa}{1-2\cdot\alpha}\right).
\end{align*}

Denote by ${v^T\rh}^*$ and ${u^T\rh}^*$ the average-sum optimal values of the under-paying $\mdp^{T}\rh$ and the over-paying $N^{T}\rh$, respectively (i.e., the expected liminf of the average value, for the best policy in $A^T$, similar to~(\ref{eq::value_rho})-(\ref{eq::value_rho_max})). The following proposition formalizes the bounds provided by the over-paying and under-paying MDPs:

\begin{proposition}\label{prop::trunerr}
For any $T\in\mathbb{N}$, if ${v\rh}^*\ge0$ then ${u^T\rh}^*\ge {v\rh}^*\ge  {v^T\rh}^*$. Moreover, these bounds are tight: $\lim\limits_{T\rightarrow\infty}{u^T\rh}^*-{v^T\rh}^*=0$.
\end{proposition}

The proof is differed to the appendix. Having introduced these MDP families, we are now ready to present an algorithm which utilizes them to obtain upper and lower bounds on the attacker's profit.

\subsection{Algorithm}
\begin{algorithm}\label{alg}\ \\
Input: $\alpha$ and $\gamma$, a truncation parameter $T_0\in\mathbb{N}$, and error parameters $0<\epsilon<8\cdot\alpha$, $0<\epsilon'<1$  
\begin{enumerate}
\item $low\leftarrow0$, $high\leftarrow1$\label{line_initseg}
\item do
\item\quad $\rho \leftarrow (low+high)/2$\label{line_initrho}
\item\quad $\left(\pi, v\right) \:\leftarrow\: mdp\_solver(\mdp\rh^{T_0},\epsilon/8)$\label{line_mdp_solver}
\item\quad if $(v>0)$\label{line_binarys1}
\item\quad\quad $low\leftarrow\rho$ \label{line_binarys2}
\item\quad else
\item\quad\quad $high\leftarrow\rho$ \label{line_binarys4}
\item while$\left(high-low\ge\epsilon/8\right)$ \label{line_while}
\item $lower$-$bound\leftarrow (\rho-\epsilon)$ \label{line::output_lowbound}
\item $lower$-$bound$-$policy\leftarrow \pi$\label{line::output_lowboundpolicy}
\item $\rho'\leftarrow \max\left\{low-\epsilon/4,0\right\}$ \label{line_low_rho_assign}
\item $\left(\pi, u\right) \:\leftarrow\: mdp\_solver(N_{\rho'}^{T_0},\epsilon')$\label{line_mdp_solver_2}
\item $upper$-$bound\leftarrow (\rho'+2\cdot(u+\epsilon'))$\label{line::output_upbound}
\end{enumerate}
\end{algorithm}

The algorithm initializes the search segment to be $[0,1]$ (line~\ref{line_initseg}) and begins a binary search: $\rho$ is assigned the middle point of the search segment (line~\ref{line_initrho}), and the algorithm outputs an $\epsilon/8$-optimal policy of $\mdp\rh^{T_0}$ and its value (line~\ref{line_mdp_solver}). The loop halts if the size of the search segment is smaller then $\epsilon/8$. Otherwise, it restricts the search to the larger half of the segment, if the value is positive (line~\ref{line_binarys2}), or to the lower half, in case it is negative~(line~\ref{line_binarys4}). This essentially represents a binary search for an approximate-root of ${v^{T_0}\rh}$, which is a monotonically decreasing function of $\rho$. The algorithm outputs $(\rho-\epsilon)$ as a lower bound on the player's relative revenue, and $\pi$ as an $\epsilon$-optimal policy. These assertions are formalized in the proposition below:

\begin{proposition}\label{prop::correctness}
For any $T_0\in\mathbb{N}$ and $\epsilon>0$, Algorithm~\ref{alg} halts, and its output $(\rho,\pi)$ satisfies: $\big|\rho-\rev(\pi)\big|<\epsilon$ and $\big|\rho-\max_{\pi'\in A^{T_0}}\left\{\rev(\pi')\right\}\big|<\epsilon$.
\end{proposition}

The second part of the algorithm (lines~\ref{line_low_rho_assign}-\ref{line::output_upbound}) computes an $\epsilon'$-optimal policy for the over-paying MDP $N^{T_0}_{\rho'}$, for $\rho'=(low-\epsilon/4)^+$ (using the value assigned last to $low$). If $u$ is the outputted value, the algorithm returns $\rho+2\cdot(u+\epsilon')$ as an upper bound to the player's revenue (line~\ref{line::output_upbound}).

\begin{proposition}\label{prop::trunerr2}
If $u$ and $\rho'$ are the outcome of the computation in Algorithm~\ref{alg}, lines~\ref{line_low_rho_assign}-\ref{line_mdp_solver_2}, then $\rho'+2\cdot (u+\epsilon') > \max_{\pi'\in A}\left\{\rev(\pi')\right\}$.
\end{proposition}

Both propositions are proved in Appendix~\ref{appendix::correctness}.

\subsection{Profit threshold Calculation}\label{subsec::calculation_threshold}
The threshold $\hat\alpha(\gamma)$ marks the minimal computational power an attacker needs in order to gain more than its fair share (see Section~\ref{sec::model}). It is crucial in assessing the system's resilience: An attacker above the threshold is able to receive increased returns on its investment, to grow steadily in resources,\textsuperscript{\ref{ftnt::growth}} and eventually to push other nodes out of the game.
The system is safe against such a destructive dynamic if all miners hold less than $\hat\alpha(\gamma)$ of the computational power.

Fix $\gamma$. A simple method allows us to lower bound the threshold: We first modify the action space of the overpaying $N_{\alpha}^T$ so as to disable the option of honest mining; technically, this is done by removing \ov from the feasible actions in $(1,0)$ and then, separately, removing \ad in $(0,1)$. Denote this modified MDP by $\widehat{N_{\alpha}^T}$. Then we solve $\widehat{N_{\alpha}^T}$, for some $\alpha$, error parameter $\epsilon$, and truncation $T$. If the $mdp\_solver$ returns a value smaller than $(-\epsilon)$ (both for when \ov is disabled in $(1,0)$ and when \ad is disabled in $(0,1)$), we are assured that honest mining is optimal in the original setup. We perform a search for the maximal $\alpha$ satisfying this requirement, i.e., $\hat{\alpha(\gamma)}$, in a fashion similar to the search in Algorithm~\ref{alg}.

\begin{corollary}\label{cor::alg_4_threshold}
Fix $\gamma$ and $\alpha$. If $u$ is the value returned by $mdp\_solver(\widehat{N_{\alpha}^T},\epsilon)$, and $u\le-\epsilon$, then \hon is optimal for $\alpha$. In other words, $\hat\alpha(\gamma)\ge\alpha$.
\end{corollary}

\section{Results}\label{sec::additional_results}
\subsection{Optimal Values}
We ran Algorithm~\ref{alg} for $\gamma$ from $\left\{0,0.5,1\right\}$, with various values of $\alpha$, using an MDP solver for MATLAB (an implementation of the relative value iteration algorithm developed by Chad{\'e}s et al.~\cite{MDPToolBOX}).
The error parameter $\epsilon$ was set to be $10^{-5}$ and the truncation was set to $T=75$. The values of $\rho$ returned by the algorithm, for $\gamma=0,0.5,1$, are depicted in Figure~\ref{fig::res_rhos} above. Additionally, some values for $\gamma=0$ appear in Table~\ref{table::revenue_gamma_0}, computed for parameters $T=95$ and $\epsilon=10^{-5}$. The results demonstrate a rather mild gap between the attacker's optimal revenue and the revenue of \sm. In addition, the graphs depict the upper bound on the revenue provided in Section~\ref{sec::upperbound}; as we stated there, the bound is obtained when $\gamma=1$, which is observed clearly in the corresponding graph.

\begin{table}[h]
\centering
\caption{\label{table::revenue_gamma_0}The revenue of the attacker under \sm and under the $\epsilon$-OPT policies, compared to the computed upper bound, for various $\alpha$ and with $\gamma=0$.}
\begin{tabular}{|c|c|c|c|}
\hline
$\alpha$ & $\sm$    & $\epsilon$-OPT & \makecell{Upper-\ \\Bound}  \\ \hline
$1/3$ & $1/3$ & $0.33705$  & $0.33707$ \\ \hline
0.35   & 0.36650 & 0.37077 & 0.37079\\ \hline
0.375  & 0.42118 & 0.42600  & 0.42604 \\ \hline
0.4   & 0.48372  & 0.48866 & 0.48904\\ \hline
0.425  & 0.55801 & 0.56808 & 0.57226\\ \hline
0.45   & 0.65177 & 0.66891 & 0.70109\\ \hline
0.475  & 0.78254 & 0.80172 & 0.90476\\ \hline
\end{tabular}
\end{table}

\subsection{Optimal Policies} Below we illustrate two examples of the behaviour of the $\epsilon$-optimal policies returned by the algorithm.
The policies are described by tables, with the row index corresponding to $\sa$ and the columns to $\sh$. The table-entry $(\sa,\sh)$ contains three characters, specifying the actions to be taken in states $(\sa,\sh,irrelevant)$, $(\sa,\sh,relevant)$, and $(\sa,\sh,active)$ correspondingly. Table~\ref{tbl::policy_0.45_0.5} contains a description of an optimal policy, for an attacker with $\alpha=0.45, \gamma=0.5$. Table~\ref{tbl::policy_1/3} describes optimal actions for the setup $\alpha=1/3, \gamma=0$.
Notice that in the latter the \ma action is irrelevant, which allows us to regard in the second table only states with $fork=irrelevant$. In both tables only a subset of the states is depicted, the whole space being infinite.

\begin{table}
\centering
\caption{The optimal policy for an attacker with $\alpha=0.45$ and $\gamma=0.5$, for states $(\sa,\sh,\cdot)$ with $\sa,\sh\le 8$. The rows index the attacker's chain length $(\sa)$, and the columns the honest network's $(\sh)$. The three characters in each entry represent the action to be taken if $fork=irrelevant$, $relevant$, or $active$. `a', `o', `m', and `w' stand for $adopt$, $override$, $match$, and $wait$, respectively, while `$*$' represents an unreachable state.}
\begin{tabular}{|c|c|c|c|c|c|c|c|c|c|}
\hline
\backslashbox{$\sa$}{$\sh$}
& 0   & 1   & 2   & 3   & 4   & 5   & 6   & 7   & 8   \\ \hline
0 & \starr \starr \starr  & \starr a\starr  & \starr \starr \starr  & \starr \starr \starr  & \starr \starr \starr  & \starr \starr \starr  & \starr \starr \starr  & \starr \starr \starr  & \starr \starr \starr  \\ \hline
1 & w\starr \starr  & \starr m\starr  & a\starr \starr  & \starr \starr \starr  & \starr \starr \starr  & \starr \starr \starr  & \starr \starr \starr  & \starr \starr \starr  & \starr \starr \starr  \\ \hline
2 & w\starr \starr  & \starr mw & \starr m\starr  & w\starr \starr  & a\starr \starr  & \starr \starr \starr  & \starr \starr \starr  & \starr \starr \starr  & \starr \starr \starr  \\ \hline
3 & w\starr \starr  & \starr mw & \starr mw & wm\starr  & w\starr \starr  & a\starr \starr  & \starr \starr \starr  & \starr \starr \starr  & \starr \starr \starr  \\ \hline
4 & w\starr \starr  & \starr mw & \starr mw & omw & wm\starr  & w\starr \starr  & w\starr \starr  & a\starr \starr  & \starr \starr \starr  \\ \hline
5 & w\starr \starr  & \starr mw & \starr mw & \starr mw & omw & wm\starr  & w\starr \starr  & w\starr \starr  & a\starr \starr  \\ \hline
6 & w\starr \starr  & \starr mw & \starr mw & \starr mw & \starr mw & omw & wm\starr  & w\starr \starr  & w\starr \starr  \\ \hline
7 & w\starr \starr  & \starr mw & \starr mw & \starr mw & \starr mw & \starr mw & ooo & w\starr \starr  & w\starr \starr  \\ \hline
8 & w\starr \starr  & \starr ww & \starr mw & \starr mw & \starr mw & \starr mw & \starr m\starr  & oo\starr  & w\starr \starr  \\ \hline
\end{tabular}
\label{tbl::policy_0.45_0.5}
\end{table}

To illustrate how Table~\ref{tbl::policy_0.45_0.5} should be read, consider entry $(\sa,\sh)=(3,3)$, for instance. The string ``${wm*}$'' in this entry reads: ``in case a fork is $irrelevant$ (that is, the previous state was $(2,3)$), $wait$; in case it is $relevant$ (the previous state was $(3,2)$), $match$; the case where a fork is already $active$ is not reachable''.

\begin{table} 
\centering
\caption{The optimal policy for an attacker with $\alpha=0.35$ and $\gamma=0$. The table describes the actions only for states of the form $(\sa,\sh,irrelevant)$ with $\sa,\sh\le 8$. (See previous caption)}
\begin{tabular}{|c|c|c|c|c|c|c|c|c|c|}
\hline
\backslashbox{$\sa$}{$\sh$}  & 0 & 1 & 2 & 3 & 4 & 5 & 6 & 7 & 8   \\ \hline
0 & \starr & a & \starr & \starr & \starr & \starr & \starr & \starr & \starr   \\ \hline
1 & w & w & w & a & \starr & \starr & \starr & \starr & \starr   \\ \hline
2 & w & o & w & w & a & \starr & \starr & \starr & \starr   \\ \hline
3 & w & w & o & w & w & a & \starr & \starr & \starr   \\ \hline
4 & w & w & w & o & w & w & w & a & \starr   \\ \hline
5 & w & w & w & w & o & w & w & w & a   \\ \hline
6 & w & w & w & w & w & o & w & w & w   \\ \hline
7 & w & w & w & w & w & w & o & w & w   \\ \hline
8 & w & w & w & w & w & w & w & o & w \\ \hline
\end{tabular}
\label{tbl::policy_1/3}
\end{table}

Looking into these optimal policies we see they differ from \sm in two ways: First, they defer using \ad in the upper triangle of the table, if the gap between \sh and \sa is not too large, allowing the attacker to ``catch up from behind''. Thus, apart from block withholding, an optimal attack may also contain another feature: attempting to catch up with the longer public chain from a disadvantage. This implies that the attacker violates the longest-chain rule, a result which counters the claim that the longest-chain rule forms a Nash equilibrium (see~\cite{FELTEN}, and discussion in Section~\ref{sec::related})

Secondly, they utilize \ma more extensively, effectively overriding the honest network's chain (w.p. $\gamma$) using one block less.

\subsection{Thresholds}\label{subsec::additional_results_threshold}
Following the method described in Section~\ref{subsec::calculation_threshold}, we are able to introduce lower bounds for the profit thresholds. Figure~\ref{fig::thresholds} depicts the thresholds induced by optimal policies, compared to that induced by \sm. The results demonstrate some cutback of the thresholds, when considering policies other than \sm.

\subsection{Evaluation of Protocol Modifications}
Several protocol modifications have been suggested to counter selfish mining attacks. It is important to provably verify the merit of such suggestions. This can be done by adapting our algorithm to the MDPs induced by these modifications.
Below we demonstrate this with respect to the rule suggested by Eyal and Sirer.
According to the Bitcoin protocol, a node which receives a chain of length equal to that of the chain it currently adopts, ought to reject the new chain. Eyal and Sirer suggest to instruct nodes to accept the new chain with probability $1/2$. We refer to it below as ``uniform tie breaking''.

The immediate effect of this modification is that it restricts the efficiency of the \ma action to $1/2$, even when the attacker's communication capabilities correspond to $\gamma>1/2$. Admittedly, this limits the power of strongly communicating attackers, and thus guarantees a positive lower bound on the threshold for profitability of \sm (which was 0, when $\gamma=1$). On the other hand, it has the apparent downside of enhancing the power of poorly communicating attackers, that is, it allows an attacker to \ma with a success-probability $1/2$ even if its ``real'' $\gamma$ is smaller than $1/2$. 

Unfortunately, our results show that this protocol enhances the profit of some attackers from deviations. For example, by applying Algorithm~\ref{alg} to the setup induced by uniform tie breaking, we found that attackers in the range $\left\{\gamma=0.5\right.,$ $\left.0.2321<\alpha<0.5\right\}$ benefit from this modification. In particular, the profit threshold deteriorates from $0.25$ to $0.2321$.
Figure~\ref{fig::thresholds5050} demonstrates this by comparing the attacker's optimal revenue under the uniform tie breaking protocol with the optimal revenue under the original protocol. The dominating policy is described in Table~\ref{table::50-50}.

The intuition behind this result is as follows: Under uniform tie breaking,  two chains of equal length will be mined equally regardless of the passage of time between the transmission of their last blocks. This allows an attacker to perform \ma even if it did not have a block prepared in advance, thereby granting it additional chances to catch up from behind. Deviation from the longest-chain rule thus becomes even more tempting.
\begin{figure}
\centering
\includegraphics[scale=0.7]{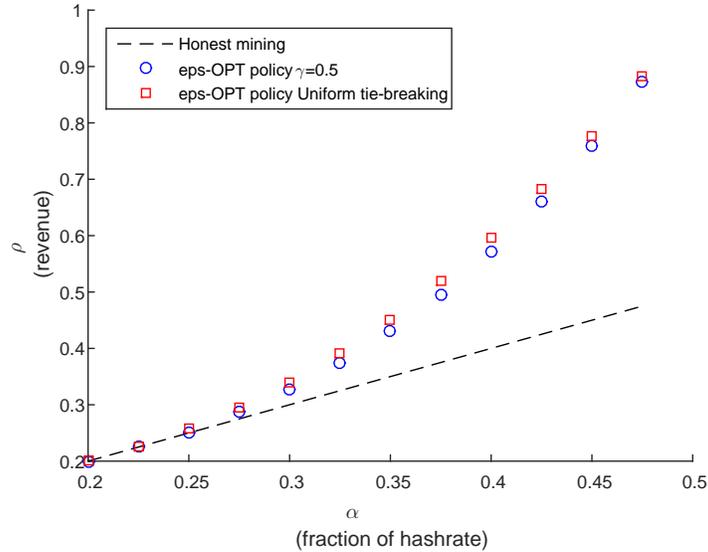}
\caption{The attacker's optimal revenue under uniform tie breaking, compared to that under the original protocol (with $\gamma=1/2$) and to \hon.}
\label{fig::thresholds5050}
\end{figure}

\begin{table}[h]
\centering
\caption{The optimal policy for an attacker with $\alpha=0.25$, under the ``50-50'' protocol modification suggested in~\cite{ES}. Only states $(\sa,\sh,\cdot)$ with $\sa,\sh\le 8$ are depicted. This policy outperforms honest mining.}
\label{table::50-50}
\begin{tabular}{|c|c|c|c|c|c|c|c|c|c|}
\hline
  & 0   & 1   & 2   & 3   & 4   & 5   & 6   & 7   & 8   \\ \hline
0 & \starr w\starr & aa\starr & \starr\starr\starr & \starr\starr\starr & \starr\starr\starr & \starr\starr\starr & \starr\starr\starr & \starr\starr\starr & \starr\starr\starr \\ \hline
1 & w\starr\starr & \starr m\starr & w\starr\starr & a\starr\starr & \starr\starr\starr & \starr\starr\starr & \starr\starr\starr & \starr\starr\starr & \starr\starr\starr \\ \hline
2 & w\starr\starr & \starr oo & m\starr\starr & w\starr\starr & a\starr\starr & \starr\starr\starr & \starr\starr\starr & \starr\starr\starr & \starr\starr\starr \\ \hline
3 & w\starr\starr & \starr w\starr & \starr oo & m\starr\starr & w\starr\starr & a\starr\starr & \starr\starr\starr & \starr\starr\starr & \starr\starr\starr \\ \hline
4 & w\starr\starr & ww\starr & \starr w\starr & \starr oo & m\starr\starr & w\starr\starr & w\starr\starr & a\starr\starr & \starr\starr\starr \\ \hline
5 & w\starr\starr & ww\starr & ww\starr & \starr w\starr & \starr oo & m\starr\starr & w\starr\starr & w\starr\starr & a\starr\starr \\ \hline
6 & w\starr\starr & ww\starr & ww\starr & ww\starr & \starr w\starr & \starr oo & m\starr\starr & w\starr\starr & w\starr\starr \\ \hline
7 & w\starr\starr & ww\starr & ww\starr & ww\starr & ww\starr & \starr w\starr & \starr oo & m\starr\starr & w\starr\starr \\ \hline
8 & w\starr\starr & ww\starr & ww\starr & ww\starr & ww\starr & ww\starr & \starr w\starr & \starr oo & m\starr\starr \\ \hline
\end{tabular}
\end{table}

\subsection{Simulations}
In order to verify the results above we built a selfish mining simulator which we implemented in Java. We ran the simulator for various values of $\alpha$ and $\gamma$ (as in the figures above), where the attacker follows the policies generated by the algorithm. Each run was performed  for $10^7$ rounds (block creation events). The relative revenue of the attacker matched the revenues returned by the algorithm, up to an error of at most $\pm10^{-6}$.

\section{A Model that Considers Delays}\label{sec::delays}
So far, our model assumed that no new block is created until all preceding published blocks arrived at all nodes. In reality, there are communication delays
between nodes in the network, including between the attacker and others. Thus, instead of modeling the attacker's communication capabilities via the parameter $\gamma$, it may be better to directly consider the non-negligible effect of network latency directly. Delays are especially noticeable when the system's throughput is increased by allowing larger blocks to form or by increasing block creation rates (see~\cite{GHOST}).
While this makes the encoding of the game rather complicated, \emph{a priori }we can make the following observations:
\begin{enumerate}
\item The attacker has only a partial knowledge of the world state. Furthermore, blocks which it publishes may arrive at the honest network too late, which potentially reduces the benefit of block withholding.
\item Natural forks occur within the honest network, and consequently its chain grows in a rate lower than one block per round; this potentially makes attacks more successful.
\item Natural forks involving the attacker imply that the game arrives at non-trivial states, even under honest mining. The attacker may thus mine honestly until some particular deviation becomes feasible.\label{point_ntstates}
\item In the presence of delays, the attacker's share under honest mining might be greater than $\alpha$, which raises the bar for dishonest strategies to prevail.\footnote{See~\cite{INCLUSIVE}, for one result quantifying this effect.}
\end{enumerate}

The overall effect of the above cannot be determined without knowing the topology of the network and the attacker's location in it (as well as its knowledge about the topology).
Still, some insight is possible.
Following the third observation above, we notice that a dishonest policy $\pi$ is one in which for some $\sh>\sa: \pi(\sa,\sh)\neq \ad$ and/or for some $\sa>\sh: \pi(\sa,\sh)\neq \ov$ (whereas under no delays honesty in $(1,0)$ and $(0,1)$ suffices).
We claim that, consequently,
the profit threshold equals 0. In other words, \emph{every} attacker benefits from some form of dishonest mining.

\begin{clm}
When the network suffers some delays, the attacker has a strict better-response strategy to honest mining, for any $\alpha>0$.
\end{clm}
Below we provide a proof sketch, which contains the jist of the claim while avoiding the involved formalization of the process under delays. We do mention that the rewards in $\mdp\rh$ are now given by the expected outcome of future events, specifically the resolution of future conflicts. For instance, following an \ov action in $(\sa,\sh)$, the attacker is awarded $(1-\rho)\cdot(\sh+1)$ times the probability that its block will be accepted by all nodes, eventually.\footnote{It can be shown that this is actually decided in finite time, in expectation~\cite{GHOST}.}
\commentin{\begin{proof}[sketch]
Fix $k\in\mathbb{N}$, and let $\pi_k$ be the policy in which the attacker mines honestly, until some state $(k-1,k)$ is reached (and observable to it). Upon which, instead of adopting, the attacker tries to catch up from behind, until it either succeeds (and then it overrides) or it learns of another block of the honest network (and then it adopts). Formally, $\pi_k(\sa,\sh):=\left\{\begin{array}{ccc}
\ov \quad \quad &  \sa>\sh    \\
\ad \quad \quad &  \sh> \sa \wedge \sh\neq k    \end{array} \right\}
$.
Since $\pi_k$ is stationary, we can analyze its long-term earnings in $\mdp\rh$ following the result from Lemma~\ref{lem::strong_law_applies}.



Denote $\rho_h:=\rev(\hon)$ and $\rho_k:=\rev(\pi_k)$.
Upon reaching $(k-1,k)$, the attacker's immediate reward under \hon is $\left(-\rho_h\cdot k\right)$. On the other hand, if it follows $\pi_k$, its expected immediate rewards are at least
\begin{equation}
\label{eq::delays}
q\cdot\left(1-\rho_k\right)\cdot(k+1)-(1-q)\cdot\rho_k\cdot(k+1),
\end{equation}
where $q$ is a lower bound on the probability that it will succeed to bypass the honest network's chain \emph{and} override it in time. 
The positive term in~(\ref{eq::delays}) corresponds to the case where the scheme ends successfully (with an \ov), and the negative one to the complementary scenario. 
To avoid dependencies on $k$, $q$ can be taken to equal
\begin{equation}
\int_{0}^{\infty}\int_{0}^{\infty}\left(\alpha\cdot\lambda\right)^2\cdot e^{-\alpha\lambda\cdot(t+s)}\cdot e^{-(1-\alpha)\cdot\lambda\cdot (t+s+d_{a,h}+d_{h,a})}dsdt,
\end{equation}
where $d_{h,a}$ is the communiation delay on the link from the honest cluster to the attacker, and $d_{a,h}$ the delay on the reversed link (for simplicity, we assume that in both directions there are single links connecting these parties to one another). Indeed, the integrand above represents the probability that the next two blocks of the attacker will take a time of $t+s$ to be generated ($\left(\alpha\cdot\lambda\right)^2\cdot e^{-\alpha\lambda\cdot(t+s)}$), and that the honest network hasn't been able to create a block since the beginning of the propagation of its $k$'th block, and until the attacker's $(k+1)$-block propagated throughout the network ($e^{-(1-\alpha)\cdot\lambda\cdot (t+s+d_{a,h}+d_{h,a})}$).

Assume by way of negation that $\rho_h\ge\rho_k$.
If $k$ is large enough, the following relation holds:
\begin{align}
& \label{eq::delay_rho_inequality1} q\cdot\left(1-\rho_k\right)\cdot(k+1)-(1-q)\cdot\rho_k\cdot(k+1)-(-\rho_h\cdot k) \ge \\ &\nonumber
q\cdot\left(1-\rho_k\right)\cdot(k+1)-(1-q)\cdot\rho_k\cdot(k+1)+\rho_k\cdot k = \\& (k+1)\cdot q - \rho_k > 0.
\label{eq::delay_rho_inequality3}
\end{align}
This implies that the expected rewards of $\pi_k$, resulting form state $(k-1,k)$ being reached, exceed those of honest mining upon reaching this state. Since this is the only state in which these strategies differ, the inequality above implies that $\pi_k$ strictly dominates honest mining, thus $\rho_k=\rev(\pi_k)>\rev(\hon)=\rho_h$. We conclude that any attacker can benefit from deviating in some states from honest mining, hence that the profit threshold vanishes.

\qed
\end{proof}}
The intuition behind this result is clear: The attacker suffers a significant loss if it adopts in $(k-1,k)$, when $k$ is large, and it thus prefers to continue the fork that formed naturally, and attempt to catch up.

This illustrates the importance of the policies found by Algorithm~\ref{alg}. As we've seen (Section~\ref{sec::additional_results}), those dominate \sm in that they delay adoption, i.e., they
allow the continuation of the attack even when the honest network's chain is longer than the attacker's. While the additional benefit was rather mild, this added feature becomes more important in networks with delays, where splits in the chain occur naturally with some probability, even when honest mining is practiced by all.

To gain further understanding of selfish mining under delays it would be important to quantify the optimal gains from such deviations.
We leave this as an open question for future research. Still, it is clear that Bitcoin will be more vulnerable to selfish mining if delays become more prominent, e.g., in the case of larger blocks (block size increases are currently being discussed within the Bitcoin developers community).

\section{Effect on Double Spending Attacks}\label{sec::doublespend}
In this section we discuss the qualitative effect selfish mining has on the security of payments. The regular operation of bitcoin transactions is as follows: A payment maker signs a transaction and pushes it to the Bitcoin network, then nodes add it to the blocks they are attempting to create. Once a node succeeds it publishes the block with its content. Although the payee can now see this update to the public chain of blocks, it still waits for it to be further extended before releasing the good or service paid for. This deferment of acceptance guarantees that a conflicting secret chain of blocks (if exists) will not be able to bypass and override the public one observed by the payee, thereby discard the transaction. Building a secret chain in an attempt to reverse payments is called a \emph{double spending} attack.

\mysubsub{Success-probability}
Satoshi Nakamoto, in his original white paper, provides an analysis regarding double spending in probabilistic terms: \emph{Given that the block containing the transaction is followed by $n$ subsequent blocks, what is the probability that an attacker with computational power $\alpha$ will be able to override this chain, now or in the future?} Nakamoto showed that the success-probability of double spending attacks decays exponentially with $n$. Alternative and perhaps more accurate analyses exist, see~\cite{MENI},\cite{GHOST}.

\mysubsub{Cost}
While a single double spending attack succeeds with negligible probability (as long as the payee waits long enough), regrettably, an attacker which \emph{continuously} executes double spending attempts will eventually succeed (\emph{a.s.}). We should therefore be more interested in the cost of an attack than in its success-probability.
Indeed, every failed double spending attack costs the attacker the potential award it could have gotten had it avoided the fork and published its blocks right away.

Observe, however, that a smart strategy for an attacker would be to continuously employ selfish mining attacks, and upon success combine them with a double spending attack. Technically, this can be done by regularly engaging in public transactions, while always hiding a conflicting one in the attacker's secret blocks.\footnote{In the worst case, the attacker is frequently engaged in ``real''  transactions anyways, hence suffers no loss from them being occasionally confirmed, when attacks fail.} There is always some probability that by the time a successful selfish mining attack has ended, the payment receiver has already accepted the payment, which additionally results in a successful double spending.

To summarize, the existence of a miner for which selfish mining is at least as profitable as honest mining fundamentally undermines the security of payments, as this attacker bears no cost for continuously attempting to double spending, and it eventually must succeed. Similarly, an attacker that cannot profit from selfish mining alone, might be profitable in the long run if it combines it with double spending, which potentially has grave implications on the profit threshold.

\section{Related Work}\label{sec::related}
The Bitcoin protocol was introduced in a white paper published in 2008 by Satoshi Nakamoto~\cite{SATOSHI}.
In the paper, Nakamoto shows that the blockchain is secure as long as a majority of the nodes in the Bitcoin network follow the protocol.
Kroll et al.~\cite{FELTEN} show that, indeed, always extending the latest block in the blockchain forms a (weak, non-unique) Nash equilibrium, albeit under a simpler model that does not account for block withholding.

On the other hand, it has been suggested by various people in the Bitcoin forum that strong nodes might be incentivized to violate the protocol by withholding their blocks~\cite{bitcoinforum}. Eyal and Sirer proved this by formalizing a block withholding strategy \sm and analyzing its performance~\cite{ES}. Their strategy thus violates the protocol's instruction to immediately publish one's blocks, but still sticks to the longest-chain rule (save a selective tie breaking). \sm1 still abandons its chain if the honest nodes create a longer chain.
One result of our paper is that even adhering to the longest-chain rule is not a best response.
We also prove what the optimal policies are, and compute the threshold under which honest mining is a (strict, unique) Nash equilibrium.
Additional work on selfish mining via block withholding appears in~\cite{bahack2013theoretical}. Transaction propagation in Bitcoin has also been analyzed from the perspective of incentives. Results in~\cite{babaioff2012bitcoin} show that nodes have an incentive not to propagate transactions, and suggests a mechanism to correct this. Additional analysis from a game theoretic perspective has also been conducted with regards to interactions pools, either from a cooperative game theory perspective~\cite{lewenberg2015bitcoin}, or when considering attacks between pools~\cite{eyal2014miner}.

A recent paper by G{\"o}bel et al. has evaluated \sm in the presence of delays~\cite{SM_Delays}. They show that \sm is not profitable under a model of delays that greatly differs from our own (in particular, they assume that block transmission occurs as a memoryless process). While \sm may indeed be unprofitable when delay is modeled, we show that other profitable selfish mining attacks exist.
Additional analysis of block creation in the presence of delays and its effects on throughput and double spending appears in~\cite{GHOST,INCLUSIVE,decker2013information}.

Further discussion on Bitcoin's stability can be found in a recent survey by Bonneau et al.~\cite{JOSEPH}.


%

\bibliographystyle{plain}
\bibliography{SelfishMining2}

\appendix

\section{Generality of the Model}\label{appendix::model}
As mentioned in Section~\ref{sec::model}, the most general setup would be for an attack-strategy
to consider also building its blocks in different places in the block-tree (say, extending a previously abandoned chain, or adopting a subchain of the public chain) and/or to publish more than $\sh+1$ blocks upon overriding the public honest chain. It is clear, intuitively, why such actions are suboptimal. Below we make this formal.

Let $\pi$ be an optimal strategy, when the above actions are available to the attacker as well.

\emph{Part I:}Assume there exists a state $(\sa,\sh)$ where the attacker publishes $\sh+j$ blocks with $j>1$; we denote this by $\pi(\sa,\sh)=$``$\ov\text{ \emph{by} }j$''. We now construct a policy $\pi'$, which follows $\pi$ everywhere except that $\pi'(\sa,\sh)=\wa$. By Corollary~\ref{cor::second_cor}, it suffices to show that $v^{\pi'}\rh\ge v^{\pi}\rh=0$, where $\rho$ is the relative revenue induced by $\pi$; this will imply that no reduction in $\rev$ occurs when switching from ``$\ov\text{ \emph{by} }j$'' to \wa.

For every state $X$, denote by $v^{\pi}\rh(X)$ the expected value of $(1-\rho)\cdot \E\left[R^{1,1}(\pi)\right]-\rho\cdot \E\left[R^{2,1}(\pi)\right]$ conditioned on arriving at state $X$ (recall that $\tau_1$ is the terminating state of the first run). We need to show that $v^{\pi'}\rh(\sa,\sh)\ge v^{\pi}\rh(\sa,\sh)$, as this is the only states where these policies differ. Observe that $v^\pi\rh(\sa,\sh)=(1-\rho)\cdot j+ v^\pi\rh(\sa-j,0)$. This is because either $j<\sa$, and this action cannot lead to a termination (hence the addition of $v^\pi\rh(\sa-j,0)$), or $j=\sa$, and then $v^\pi\rh(\sa-j,0)=v^\pi\rh(0,0)=0$, which fits the fact that a termination occurred. We thus need to show that $v^{\pi'}\rh(\sa,\sh)\ge(1-\rho)\cdot j+ v^\pi\rh(\sa-j,0)$.

Indeed, consider the case where $\pi'$ performs ``$\ov\text{ \emph{by} }j$'' if $X=(\sa,\sh+1)$ or ``$\ov\text{ \emph{by} }(j+1)$'' if $X=(\sa+1,\sh)$. Note that the action in the first case is feasible, since $\sa\ge\sh+j>\sh+1$, and similarly $\sa+1\ge\sh+j+1>\sh$, for the second case. In the former case we obtain $v^{\pi'}\rh(\sa,\sh+1)=(1-\rho)\cdot j+ v^\pi\rh(\sa-j,0)$, and in the latter, $v^{\pi'}\rh(\sa,\sh+1)=(1-\rho)\cdot (j+1)+ v^\pi\rh(\sa+1-(j+1),0)$. Therefore, we have presented an action-scheme which guarantees $\pi'$ the value of $\pi$. As $\pi$ (hence $\pi'$) optimize the value $v^{\pi}\rh(X)$, for any $X$, we have that the value of $\pi$ (hence of $\pi'$) in the states $(\sa+1,\sh)$ and $(\sa,\sh+1)$ is at least as high as $(1-\rho)\cdot j+ v^\pi\rh(\sa-j,0)$, which completes this part of the proof.

\emph{Part II:}
We claimed, additionally, that the attacker will never adopt branches in the block-tree other than its current secret one and the honest faction's current longest one. We now aim to justify this assertion, albeit with some informalities; a formal proof is not possible under our model, because it implicitly assumes that actions as \ov and \ad grant immediate reward, whereas if the attacker adopts older abandoned chains it can hypothetically reverse such decisions. Nonetheless it is very clear why this would be suboptimal:

For any $(\sa,\sh)$, let $A_1,...,A_{\sa}$ denote the attacker's chain, and $H_1,...,H_{\sh}$ the honest network's chain, and let $H_0$ be the block that $A_1$ and $H_1$ extend (it is now public, but may have belonged to the attacker).\footnote{In case the honest network is forked, pick one of them arbitrarily; blocks are anonymous, and they are only accepted or rejected according to the lengths of their chains, which are in this case equal.} Let now $(\sa,\sh)$ be the first state at which the attacker decides to deviate and extend a block $B$ other than $A_{\sa}$ or $H_{\sh}$. If $B$ was not created after $H_0$ (and $B\neq H_0$), then it was available to the attacker at the time it began extending $H_0$. By the choice of $(\sa,\sh)$, extending $H_0$ was then at least as profitable as extending $B$, and this dominance is invariant under future events (e.g., by the public chain that formed above $H_0$). Thus the attacker can just as well repeat its initial choice of $H_0$ over $B$.

A similar argument holds for the case where $B$ was created after $H_0$ (or $B=H_0$). Denote by $l$ the length of the attacker's chain upon the creation of $B$. Extending $A_{l}$ was then at least as profitable as extending $B$, by the choice of $(\sa,\sh)$, and this again is not altered by future events. All the same, the attacker can just as well repeat its choice and choose $A_{l}$ over $B$. in conclusion, we can restrict our attention to strategies restricted to our three-action model (four, with \wa), without loss of generality. This also enables a Markovian model, fortunately, as described in Section~\ref{sec::model}.

\section{Proof of Proposition~\ref{prop::upperbound}} \label{appendix::upperbound}
\noindent {\bf Proposition~\ref{prop::upperbound}:}\\
\emph{For any $\pi$, $\rev(\pi,\alpha,\gamma)\le\frac{\alpha}{1-\alpha}$. Moreover, this bound is tight, and achieved when $\gamma=1$.
}
\commentin{\begin{proof}
We can map every block of the honest network which was overridden, to a block of the attacker; this is because \ov requires the attacker to publish a chain longer than that of the honest network's.

Let $k_T$ be the number of blocks that the attacker has built up to time $T$. The honest network thus built $l_T:=T-k_T$ by this time. The argument above shows that $l_T-\sum_{t=1}^Tr^2_t\le k_T$. Also, $\Pr(l_T>k_T)\rightarrow1$, when $T\rightarrow \infty$. Therefore, the relative revenue satisfies:
\begin{align}
& \label{eq::revbound1}\rev(\pi)=\lim\limits_{T\rightarrow\infty}\frac{\sum_{t=1}^Tr^1_t}{\sum_{t=1}^Tr^1_t+\sum_{t=1}^Tr^2_t}\le \lim\limits_{T\rightarrow\infty}\frac{\sum_{t=1}^Tr^1_t}{\sum_{t=1}^Tr^1_t+l_T-k_T} = \\&
\lim\limits_{T\rightarrow\infty}\frac{1}{1+\left(l_T-k_T\right)/\left(\sum_{t=1}^Tr^1_t\right)} \le  \lim\limits_{T\rightarrow\infty}\frac{1}{1+\left(l_T-k_T\right)/k_T} = \lim\limits_{T\rightarrow\infty} \frac{k_T}{l_T}.\label{eq::revbound2}
\end{align}
The SLLN applies naturally to $k_T$ and $l_T$, implying that the above equals $\frac{\alpha\cdot T}{(1-\alpha)\cdot T} = \frac{\alpha}{1-\alpha}$ (\emph{a.s.}).

To see that the bound is achieved in $\gamma=1$, observe that the policy SM1 satisfies the property that every block of the attacker overrides one block of the honest network, and that none of the attacker's blocks are overridden (as the policy never reaches a state where it needs to $adopt$, except when $\sa=0$). This turns both inequalities in~(\ref{eq::revbound1})-(\ref{eq::revbound2}) into equalities.
\qed\end{proof}
}

\section{Correctness of Algorithm~\ref{alg}}\label{appendix::correctness}

In this section we prove that Algorithm~\ref{alg} halts and that its output meets the conditions specified therein.
We begin with applying here a Strong Law of Large Numbers, which will prove useful along our path.
Under a fixed stationary policy $\pi$, we denote by $\tau_1$ the renewal time of the game. Formally, $\tau_1$ is the time, or number of visited states, until the game reaches a state $s$ from which the transition probabilities are $\alpha$ to state $(1,0)$ and $1-\alpha$ to $(0,1)$.
\begin{lemma}\label{lem::strong_law_applies} Let $\pi$ be some fixed policy of $\mdp\rh^{T_0}$. Denote $R^{k,1}(\pi)=\sum_{t=1}^{\tau_1} r^k_t(\pi)$ (for $k=1,2$).
\begin{align}
&\lim\limits_{T\rightarrow\infty}\frac1T\sum_{t=1}^T r^k_t(\pi)=  \E\left[\lim\limits_{T\rightarrow\infty}\frac1T\sum_{t=1}^T r^k_t(\pi)\right] = \frac{\E\left[R^{k,1}(\pi)\right]}{\E[\tau_1]}\quad (a.s.)\label{eq::strong_law_applies2},
\end{align}
for $k=1,2$. Similarly,
\begin{align}
&\lim\limits_{T\rightarrow\infty}\frac1T\sum_{t=1}^T w_{\rho}(r_t(\pi))=  \E\left[\lim\limits_{T\rightarrow\infty}\frac1T\sum_{t=1}^T w_{\rho}(r_t(\pi))\right]  \\ & \qquad\qquad=\frac{(1-\rho)\cdot \E\left[R^{1,1}(\pi)\right]-\rho\cdot\E\left[R^{2,1}(\pi)\right]}{\E[\tau_1]} \quad (a.s.)\label{eq::strong_law_applies}
\end{align}
\end{lemma}
\commentin{\begin{proof} 
Define by $\mathcal{C^\pi}$ the states reachable from state $s_0:=(1,0,irrelevant)$, when $\pi$ is employed. We will show that $\mathcal{C^\pi}$ is an irreducible positive recurrent Markov chain.
For any state $X$ in $\mathcal{C^\pi}$ it must be that the waiting time for the next visit of $s_0$ has finite expectation: Assume that after $T'$ steps the honest network created $M(T')$ blocks and the attacker $m(T')$. If $M(T')>m(T')$ then as long as the player does not adopt $\sh-\sa=M(T')-m(T')$; this is regardless of other actions which the attacker possibly made in the past. As block creations are \emph{i.i.d}, the process $Y(T')=M(T')-m(T')$ is equivalent to a random walk on $\mathbb{Z}$ with a positive drift, hence the expected time of the \emph{last} time it returns to the origin is finite. After which the only action the attacker can make is \ad and \wa. As our model does not allow for pathological strategies in which the attacker waits for periods of infinite expected length, the next adoption occurs in finite expected time.
Finally, every adoption leads to $X_0$ with probability $\alpha$, thus the next return to $X_0$ is of finite expectation. This state is thus positive recurrent.
We conclude that $\mathcal{C^\pi}$ consists of a single communicating class (the finite expectation of the return implies the existence of a $t$ for which there's a positive probability to return to $s_0$ within $t$ steps), hence that $\pi$ 
induces a single irreducible Markov chain $\mathcal{C^\pi}$, which is also positive recurrent, as $s_0$ is.
We can thus use The Strong Law of Large Numbers for Markov chains (see, e.g.,~\cite{SERFOZO} pg. 50, Corollary 79
) to arrive at~(\ref{eq::strong_law_applies2}) and~(\ref{eq::strong_law_applies}). The right-hand side equality in~(\ref{eq::strong_law_applies2}) follows from the SLLN applied to renewal reward processes. 
\qed\end{proof}
}

		\commentout{see serfozo2009basics.
		 , pg 50, corollary 79. Theorem~8.10 in http://web.stanford.edu/class/cme308/OldWebsite/notes/chap8.pdf.}
The following are immediate corollaries of the strong law above:
\begin{corollary}\label{cor::expectation_commute_limit}
For any admissible policy $\pi$ of $\mdp\rh^{T_0}$,
\begin{align}&
\label{eq::expectation_commute_limit}
& \rev(\pi)= \frac{\lim\limits_{T\rightarrow\infty}\frac1T\sum_{t=1}^Tr^1_t}{\lim\limits_{T\rightarrow\infty}\frac1T\sum_{t=1}^T\left(r^1_t+r^2_t\right)} = \frac{\E\left[R^{1,1}\right]}{\E\left[R^{1,1}\right]+\E\left[R^{2,1}\right]} 
\quad (a.s.)
\end{align}
\end{corollary}
\begin{corollary}\label{cor::second_cor}
Let $\pi$ and $\pi'$ be two policies.
\begin{enumerate}
\item If $(1-\alpha)\cdot \E\left[R^{1,1}\right]-\alpha\cdot \E\left[R^{2,1}\right]\ge0$, then $\pi$ dominates honest-mining.
\item If $(1-\rev(\pi))\cdot \E\left[R^{1,1}(\pi')\right]-\rev(\pi)\cdot \E\left[R^{2,1}(\pi')\right]>0$, then $\pi'$ dominates $\pi$.
\end{enumerate}
Both assertions become strict together with the inequalities.
\end{corollary}
	
The following lemma states that an optimal policy in $\mdp^{T_0}\rh$, whose value is small enough, is approximately optimal in $\mdp$, if only truncated policies are considered:
\begin{lemma}\label{lem::slope_of_rev}
Let $\rho\in[0,1]$, $\epsilon>0$, and $T_0\in\mathbb{N}$. If $\pi\in A^{T_0}$ is optimal in $\mdp^{T_0}\rh$ and $|v^\pi\rh|<\epsilon/2$, then
\begin{enumerate}
\item $\big|\rho-\rev(\pi)\big|<\epsilon$ \label{eq::rho_approx_rev}
\item $\big|\rho-\max_{\pi'\in A^{T_0}}\left\{\rev(\pi')\right\} \big|<\epsilon$ \label{eq::rho_approx_max}
\end{enumerate}
\end{lemma}
\commentin{\begin{proof}
Observe that $\lim\limits_{T\rightarrow\infty}\frac1T\sum_{t=1}^T\left(r^1_t\left(\pi'\right)+ r^2_t\left(\pi'\right)\right)$ represents the average number of blocks added to the agreed pubic chain (aka main chain), per round, when $\pi'$ is deployed. Under the honest strategy, this rate equals 1, as every round accounts for the addition of a new block (see Section~\ref{sec::model}). On the other hand, no positive recurrent
strategy can more than halve the growth rate of the main chain: For every block that is overridden and excluded from the main chain there's a corresponding overriding block is included in it (see also the proof of Proposition~\ref{prop::upperbound}).\footnote{This assumption is without loss of generality, as at some point the player would need to adopt, and the waiting time for it is finite in expectation. See the proof of Lemma~\ref{lem::strong_law_applies}.\label{footnote::positiverecurrent}}
Thus, $\E\left[\lim\limits_{T\rightarrow\infty}\frac1T\sum_{t=1}^T\left(r^1_t\left(\pi'\right)+ r^2_t\left(\pi'\right)\right)\right]\ge 1/2$.

\noindent \emph{Part I:} Relying on Lemma~\ref{lem::strong_law_applies} we can manipulate the limits to obtain
\begin{align}\label{eq::value_and_rev}
&\epsilon/2>v^\pi\rh = \E\left[\liminf\limits_{T\rightarrow\infty}\frac1T\sum_{t=1}^T w_{\rho}(r^1_t\left(\pi\right),r^2_t\left(\pi\right))\right] = \\ \nonumber &\E\left[ \liminf\limits_{T\rightarrow\infty}\frac1T\sum_{t=1}^T(1-\rho)\cdot r^1_t\left(\pi\right)-\rho\cdot r^2_t\left(\pi\right)\right] = \\\nonumber  &
\E\left[\lim\limits_{T\rightarrow\infty}\frac1T\sum_{t=1}^T r^1_t\left(\pi\right)\right]-\rho\cdot\E\left[\lim\limits_{T\rightarrow\infty}\frac1T\sum_{t=1}^T\left(r^1_t\left(\pi\right)+ r^2_t\left(\pi\right)\right)\right].
\end{align}
Using Corollary~\ref{cor::expectation_commute_limit} we obtain
\begin{align*}
&\rev(\pi)=\frac{\E\left[\lim\limits_{T\rightarrow\infty}\frac1T\sum_{t=1}^T r^1_t\left(\pi\right)\right]}{\E\left[\lim\limits_{T\rightarrow\infty}\frac1T\sum_{t=1}^T\left(r^1_t\left(\pi\right)+ r^2_t\left(\pi\right)\right)\right]} < \\ \nonumber&
\rho+\frac{\epsilon/2}{\E\left[\lim\limits_{T\rightarrow\infty}\frac1T\sum_{t=1}^T\left(r^1_t\left(\pi\right)+ r^2_t\left(\pi\right)\right)\right]}  \le \rho+\epsilon.
\end{align*}

Similarly, $v^\pi\rh >-\epsilon/2$ implies
\begin{align*}
&\rev(\pi) > \rho-\frac{\epsilon/2}{\E\left[\lim\limits_{T\rightarrow\infty}\frac1T\sum_{t=1}^T\left(r^1_t\left(\pi\right)+ r^2_t\left(\pi\right)\right)\right]}  \ge \rho-\epsilon,
\end{align*}

which concludes the first part.

\noindent \emph{Part II:} We use here the same technique as previously. Assume by negation that for some policy $\pi'\in A^{T_0}$, $\rev(\pi')\ge \rho+\epsilon$.
Then, similar to the previous article, we have
\begin{align}
\label{eq::epsilon_half0}
&\frac{\E\left[\lim\limits_{T\rightarrow\infty}\frac1T\sum_{t=1}^T r^1_t\left(\pi'\right)\right]}{\E\left[\lim\limits_{T\rightarrow\infty}\frac1T\sum_{t=1}^T\left(r^1_t\left(\pi'\right)+ r^2_t\left(\pi'\right)\right)\right]} =\rev(\pi')\ge \rho+\epsilon \Longrightarrow \\
&v^{\pi'}\rh=\E\left[\lim\limits_{T\rightarrow\infty}\frac1T\sum_{t=1}^T r^1_t\left(\pi'\right)\right]-\rho\cdot\E\left[\lim\limits_{T\rightarrow\infty}\frac1T\sum_{t=1}^T\left(r^1_t\left(\pi'\right)+ r^2_t\left(\pi'\right)\right)\right]\ge
\nonumber\\
&\left(\E\left[\lim\limits_{T\rightarrow\infty}\frac1T\sum_{t=1}^T\left(r^1_t\left(\pi'\right)+ r^2_t\left(\pi'\right)\right)\right]\right)\cdot\epsilon \ge 1/2\cdot\epsilon > v^\pi\rh,
\label{eq::epsilon_half2}
\end{align}
which contradicts the optimality of $v^\pi\rh$. This proves that $\rho>\max_{\pi'\in A^{T_0}}\left\{\rev(\pi)\right\}$ $-\epsilon$. On the other hand, assume in negation that $\rev(\pi)\le\rho-\epsilon$. We then have,
\begin{align*}
&\frac{\E\left[\lim\limits_{T\rightarrow\infty}\frac1T\sum_{t=1}^T r^1_t\left(\pi\right)\right]}{\E\left[\lim\limits_{T\rightarrow\infty}\frac1T\sum_{t=1}^T\left(r^1_t\left(\pi\right)+ r^2_t\left(\pi\right)\right)\right]} \le \rho-\epsilon \Longrightarrow \\
& v^{\pi}\rh=\E\left[\lim\limits_{T\rightarrow\infty}\frac1T\sum_{t=1}^T r^1_t\left(\pi\right)\right]-\rho\cdot\E\left[\lim\limits_{T\rightarrow\infty}\frac1T\sum_{t=1}^T\left(r^1_t\left(\pi\right)+ r^2_t\left(\pi\right)\right)\right]\le
\\
&\left(\E\left[\lim\limits_{T\rightarrow\infty}\frac1T\sum_{t=1}^T\left(r^1_t\left(\pi\right)+ r^2_t\left(\pi\right)\right)\right]\right)\cdot-\epsilon \le 1/2\cdot(-\epsilon) < v^\pi\rh,
\end{align*}
and we arrive again at a contradiction. Therefore, $\rev(\pi)>\rho-\epsilon$, hence $\max_{\pi'\in A^{T_0}}\left\{\rev(\pi')\right\}>\rho-\epsilon$.
\qed\end{proof}
}
\begin{corollary}\label{cor::almost_opt}
If $\pi$ is $\epsilon/4$-optimal in $\mdp^{T_0}\rh$ and $|v^\pi\rh|<\epsilon/4$, then the inequalities guaranteed by Lemma~\ref{lem::slope_of_rev} hold. 
\end{corollary}
\commentin{\begin{proof}
The first inequality holds for $\pi$, as in its proof we didn't use the assumption on $\pi$'s optimality. The second inequality is a property of $\rho$ (and not of the policy); it holds because $|v^\pi\rh|<\epsilon/4$ together with $\pi$ being $\epsilon/4$-optimal imply $|v^{\hat{\pi}}\rh|<\epsilon/2$, for an optimal policy $\hat{\pi}$.
\qed\end{proof}
}

Finally, we are ready to prove the correctness of Algorithm~\ref{alg}:\\
\noindent {\bf Proposition~\ref{prop::correctness}:}\\
\emph{For any $T_0\in\mathbb{N}$ and $\epsilon>0$, Algorithm~\ref{alg} halts, and its output $(\rho,\pi)$ satisfies: $\big|\rho-\rev(\pi)\big|<\epsilon$ and $\big|\rho-\max_{\pi'\in A^T}\left\{\rev(\pi')\right\}\big|<\epsilon$.
}
\commentin{\begin{proof}
Observe that ${v\rh^T}^*$, the optimal value of $\mdp^{T_0}\rh$, is monotonically decreasing in $\rho$: If $\rho_1>\rho_2$ and $\pi_1$ is optimal in $\mdp^{T_0}_{\rho_1}$, then ${v^{T_0}_{\rho_2}}^*\ge v^{T_0,\pi_1}_{\rho_2}>v^{T_0,\pi_1}_{\rho_1}={v^{T_0}_{\rho_1}}^*$, where the strict inequality holds because $w\rh$ is strictly decreasing. 
Furthermore, ${v^{T_0}\rh}^*$ is continuous in $\rho$, as $w\rh$ is.

Now, the quantity $(high-low)$ is halved at every iteration of the loop (lines~(\ref{line_binarys2}),(\ref{line_binarys4})), hence the number of iterations must be finite. To understand what we can say about $v$ when the algorithm halts and  $high-low<\epsilon/8$, we make use of loop invariants: First, we claim that for every value assigned to $low$ throughout the algorithm's run, the value returned by $mdp\_solver(\mdp_{low}^{T_0},\epsilon/8)$ is positive. Indeed, $low$ begins with a value of $0$. Honest mining gains the attacker a value of $\alpha$, in $\mdp_{0}^{T_0}$; $mdp\_solver(\mdp_{low}^{T_0},\epsilon/8)$ thus returns a positive value, assuming $\epsilon<8\cdot\alpha$. Any further alteration of $low$'s value, in line~\ref{line_binarys2}, is conditioned to satisfy this assertion.

Similarly, the value returned by $mdp\_solver(\mdp_{1}^{T_0},\epsilon/8)$ must be negative, since the attacker's profits for its blocks vanishes, and its revenue for blocks it adopts is negative (and such events occur in finite time, in epxectation; see Lemma~\ref{lem::strong_law_applies}). In addition, any new assignment to $high$ is conditioned to be non-positive, by line~\ref{line_binarys4}.

From the monotonocity and continuity of ${v^{T_0}\rh}^*$ we deduce that the root of ${v^{T_0}\rh}^*$ lies between $low$ and $high$. However, $high-low<\epsilon/8$ implies that $|v|<\epsilon/8$: Indeed, assume in negation that $v^\pi\rh\ge\epsilon/8$. Then
\begin{align}
& \epsilon/8\le v^{\pi}\rh=\E\left[\lim\limits_{T\rightarrow\infty}\frac1T\sum_{t=1}^T r^1_t\left(\pi\right)\right]-\rho\cdot\E\left[\lim\limits_{T\rightarrow\infty}\frac1T\sum_{t=1}^T\left(r^1_t\left(\pi\right)+ r^2_t\left(\pi\right)\right)\right]\le
\nonumber\\
& \E\left[\lim\limits_{T\rightarrow\infty}\frac1T\sum_{t=1}^T r^1_t\left(\pi\right)\right]-\left(\rho+\epsilon/8\right)\cdot\E\left[\lim\limits_{T\rightarrow\infty}\frac1T\sum_{t=1}^T\left(r^1_t\left(\pi\right)+ r^2_t\left(\pi\right)\right)\right] +\epsilon/8\\\nonumber &
\le {v^{T_0}_{\rho+\epsilon/4}}^* +\epsilon/8< {v^{T_0}_{high}}^*+\epsilon/8.
\end{align}
We used here the inequality $\E\left[\lim\limits_{T\rightarrow\infty}\frac1T\sum_{t=1}^T\left(r^1_t\left(\pi\right)+ r^2_t\left(\pi\right)\right)\right]\le 1$ (see the proof of Lemma~\ref{lem::slope_of_rev}), and the strict monotonicity of ${v^{T_0}\rh}^*$. This contradicts ${v^{T_0}_{high}}^*\le 0$. A similar derivation rules out the case $v^\pi\rh\le-\epsilon/8$, which holds as a loop invariant. We conclude that $|v|<\epsilon/8$.
\qed\end{proof}
}

\noindent{\bf Proposition~\ref{prop::trunerr}.}
\emph{For any $T\in\mathbb{N}$, if ${v\rh}^*\ge0$ then ${u^T\rh}^*\ge {v\rh}^*\ge  {v^T\rh}^*$. Moreover, these bounds are tight: $\lim\limits_{T\rightarrow\infty}{u^T\rh}^*-{v^T\rh}^*=0$.
}

We precede the proof of the proposition with some (fun!) probability analysis.
Denote by $\lt$ the set of states $\left\{(\sa,\sh) : \sa\ge\sh \right\}$. Fix some policy $\pi$ and a state $(\sa_0,\sh_0)$. Denote by $Y_t^\pi$ the random process defined by our game, where the initial state is $(\sa_0,\sh_0)$. Let $\psi$ be a stopping time defined by $\max\left\{t \;:\; Y_t^\pi\in\lt \right\}$. If $Y_{\psi}^\pi=(\sa_1,\sa_1)$ (observe that $Y_{\psi}^\pi$ must lie in the main diagonal), we denote $last(\sa_0,\sh_0):=\sa_1$. $last(\sa_0,\sh_0)$ represents the number of blocks the attacker (or the honest network, for the matter) has, before leaving $\lt$ for the rest of the epoch.

\begin{lemma}
For any state $(\sa,\sh_0)\in\lt$,
\begin{align}
\label{eq::return_time}\E\left[last(\sa_0,\sh_0)\right]&=\frac{\alpha\cdot(1-\alpha)}{\left(1-2\cdot\alpha\right)^2}+\frac12\cdot\left(\frac{\sa_0-\sh_0}{1-2\cdot\alpha}+\sa_0+\sh_0\right)
\end{align}
\end{lemma}
\begin{proof}
Note first that $last(\sa_0,\sh_0)=\frac12\cdot\left(\psi-(\sa_0-\sh_0)\right)+\sa_0$, because if the attacker created $k$ blocks after reaching $(\sa_0,\sh_0)$, the honest network needs to create precisely $k+\sa_0-\sh_0$ blocks in order to leave $\lt$. We are thus left with the task of calculating $\E\left[\psi\right]$.
Consider a random walk on $\mathbb Z$, starting at $\sa_0-\sh_0$, with probability $\alpha$ of moving one step towards positive infinity and $(1-\alpha)$ of moving towards negative infinity. Let $\psi'$ be the time until the last visit of the origin. Observe that $\psi'$ has the same distribution as $\psi$ (!), we thus identify them with each other, henceforth.

We further break $\psi$ into stopping times: Let $N$ be the number of visits to the origin (we have $N>0$ \emph{almost surely}, since the drift is towards negative infinity). Let $\psi_1$ be the first time up to the first visit of the origin, and for $1<k\le N$, let $\psi_k$ be the time that elapsed between $\psi_{k-1}$ and the next visit to the origin.
Any two travels that begin and end at the origin are i.i.d, and, moreover, the number of such travels is independent of their lengths. Therefore, by Wald's equation, $\E\left[\psi-\psi_1\right]=\E\left[N\right]\cdot\E\left[\psi_2\right]$.

We can interpret $N$ as counting the number of failures before one success, where a success represents a visit of the origin which never returns to it (this is equivalent, \emph{almost surely}, to never returning to the nonnegative side of $\mathbb Z$). The probability of a success is $\left(1-\frac{\alpha}{1-\alpha}\right)$, implying that $\E\left[N\right]=\frac{\frac{\alpha}{1-\alpha}}{1-\frac{\alpha}{1-\alpha}}=\frac{\alpha}{1-2\cdot\alpha}$.

Whenever the walk starts at $+1$ the expected return time to the origin is $\frac{1}{1-2\cdot\alpha}$. The same expression holds for the expected return time when starting at $-1$, conditioned on a return occurring (see~\cite{STERN}). Counting the first step to $\pm 1$ as well, the expected next return to the origin, conditioned on its occurrence, is $\left(1+\frac{1}{1-2\cdot\alpha}\right)$.\footnote{Note that, starting at $+1$, the expected return time is unaffected by conditioning on an eventual return, since this occurs w.p.1.} We conclude that $\E\left[\psi-\psi_1\right] = \frac{\alpha}{1-2\cdot\alpha}\cdot \left(1+\frac{1}{1-2\cdot\alpha}\right)$.

Another result in~\cite{STERN} implies that $\E\left[\psi_1\right]=\frac{\sa_0-\sh_0}{1-2\cdot\alpha}$. We obtain:
\begin{align}
\label{eq::psiExpec}
& \E\left[\psi\right]=\frac{\alpha}{1-2\cdot\alpha}\cdot \left(1+\frac{1}{1-2\cdot\alpha}\right)+\frac{\sa_0-\sh_0}{1-2\cdot\alpha} \Longrightarrow \\
& \E\left[last(\sa_0,\sh_0)\right]=\frac12\cdot\left(\E\left[\psi \right]-(\sa_0-\sh_0)\right)+\sa_0 = \frac12\cdot\left(\E\left[\psi \right]+\sa_0+\sh_0\right)=\\ &
\frac12\cdot\left(\frac{\alpha}{1-2\cdot\alpha}\cdot \left(1+\frac{1}{1-2\cdot\alpha}\right)+\frac{\sa_0-\sh_0}{1-2\cdot\alpha}+\sa_0+\sh_0\right) = \\ &
\frac{\alpha\cdot(1-\alpha)}{\left(1-2\cdot\alpha\right)^2}+\frac12\cdot\left(\frac{\sa_0-\sh_0}{1-2\cdot\alpha}+\sa_0+\sh_0\right).
\end{align}
\qed
\end{proof}

\begin{proof}[of Proposition~\ref{prop::trunerr}]
\noindent{\emph{Part I:}}
Let $\pi$ be an optimal policy in $\mdp\rh$. Assume the game has reached state $(\sa,\sh)\in\lt$, and an oracle lets the attacker know that this is the last state in $\lt$ which the game will reach before a future \ad. Assume further that the oracle lets the attacker ``cheat'' and perform the action \ma with success probability 1 (granting him, effectively, $\gamma=1$) and even if the previous state was not in $\lt$ (ignoring thus restrictions on the feasibility of \ma). Obviously, the attacker can only benefit from this oracle, by waiting for the last state in $\lt$, and then performing \ma (it has nothing to lose by taking only the null action up to that point).

Upon which, performing \ma on the main diagonal marks the end of the first epoch, since the respective chains of the attacker and the honest network collapse, hence the next state is distributed as $X_0$ is. As a result, we may bound the accumulated immediate rewards from state $(\sa,\sh)\in\lt$ onwards, up to $\tau_1$, in $\mdp\rh$, by
\begin{equation}
\label{eq::overriderevU}
(1-\rho)\cdot\E\left[last(\sa,\sh)\right]=(1-\rho)\cdot\frac{\alpha\cdot(1-\alpha)}{\left(1-2\cdot\alpha\right)^2}+\frac12\cdot\left(\frac{\sa-\sh}{1-2\cdot\alpha}+\sa+\sh\right).
\end{equation}
This is precisely the reward given in state $(\sa,\sh)\in\lt$ with $\sa=T$, in the over-paying MDP $N^T\rh$.

We follow the same approach to bound the accumulated rewards from states $(\sa,\sh)\notin\lt$. Assume that the oracle tells the attacker whether it will ever return to the main diagonal (without adopting first) or not. Clearly, if the oracle carries the negative message, the attacker is better off adopting right away, minimizing its negative reward.\footnote{Recall it is forced to \ad at \emph{some }stage, as we've seen before.} This will imply a reward of $-\rho\cdot\sh$. On the other hand, if the oracle says the process will eventually return to the main diagonal, the attacker is better off waiting for that event. If we denote by $(\sa_0,\sa_a)$ the next arrival at $\lt$ (which is necessarily on the main diagonal), then $\E\left[\sa_0 | \text{return occurs}\right]=\frac{\sh-\sa}{1-2\cdot\alpha}$~(\cite{STERN}).

Upon which the attacker's future rewards up to $\tau_1$ are bounded from above by $(1-\rho)\cdot\E\left[last(\sa_0,\sa_0)\right]=(1-\rho)\cdot\left(\frac{\alpha\cdot(1-\alpha)}{\left(1-2\cdot\alpha\right)^2}+\sa_0\right)$, by~(\ref{eq::return_time}). Since this is linear in $\sa_0$, we conclude that the expected reward from state $(\sa,\sh)\notin\lt$, conditioned on returning to the $\lt$, is upper bounded by $(1-\rho)\cdot\left(\frac{\alpha\cdot(1-\alpha)}{\left(1-2\cdot\alpha\right)^2}+\frac{\sh-\sa}{1-2\cdot\alpha}\right)$. The probability of this event is $\left(\alpha/(1-\alpha)\right)^{\sh-\sa}$. All in all, the attacker's rewards from state $(\sa,\sh)\notin\lt$ onward are upper bounded by
\begin{align}\label{eq::adoptrevU}
&\left(1-\left(\frac{\alpha}{1-\alpha}\right)^{\sh-\sa}\right)\cdot\left(-\rho\cdot\sh\right)+\\&\nonumber\left(\frac{\alpha}{1-\alpha}\right)^{\sh-\sa}\cdot(1-\rho)\cdot\left(\frac{\alpha\cdot(1-\alpha)}{\left(1-2\cdot\alpha\right)^2}+\frac{\sh-\sa}{1-2\cdot\alpha}\right).
\end{align}
This, again, is exactly the reward given in the over-paying $N^T\rh$ when state $(\sa,\sh)\notin\lt$ with $\sh=T$ is reached.

\noindent{\emph{Part II:}}
Recall the result of Lemma~\ref{lem::strong_law_applies}:
\begin{align}
& v^{\pi}\rh= \frac{(1-\rho)\cdot \E\left[R^{1,1}(\pi)\right]-\rho\cdot\E\left[R^{2,1}(\pi)\right]}{\E[\tau_1]}.
\label{eq::stronglaw3}\end{align}
When the optimal (in $\mdp\rh$) $\pi$ is applied in $N^T\rh$, with an \ad in the truncating states, the expected epoch time cannot be greater. Therefore, if ${v^\pi\rh}={v^*\rh}\ge0$, this transformation can only increase ${v^\pi\rh}$. We conclude that if $\pi$ is optimal policy in $\mdp\rh$, then the expected average value of (the truncated version of) $\pi$ in $N^T\rh$ upper bounds $v^\pi\rh=v^*\rh$. An optimal policy of $N^T\rh$ can only do better, hence ${u^T\rh}^*\ge v^*\rh$, which concludes the involved part of the proof.

\noindent{\emph{Part III:}}
That $v^*\rh\ge{v^T\rh}^*$ is trivial, since any policy that is feasible in $\mdp^T\rh$ is feasible in $\mdp\rh$, and the rewards are identical. Finally, we show that ${u^T\rh}^*\searrow v^T\rh$. First, observe that the reward from visiting a state $(\sa,\sh)\notin\lt$, given in~(\ref{eq::adoptrevU}), converges to $\rho\cdot\sh$. Thus, as $T$ goes to infinity, the reward from these states in $N^T\rh$ converges to that of $\mdp^T\rh$. On the other hand, the probability to reach a state in $\lt$ vanishes exponentially with $T$ (e.g., by applying Chernoff's bound). The reward given in $N^T\rh$ in the truncating states of the form $(T=\sa\ge\sh)$ grows only linearly in $T$ (see~(\ref{eq::overriderevU}); $\sa$ and $\sh$ are linear in $T$). Therefore, the expected reward from these states (without conditioning on reaching them) vanishes. We thus obtain $\left({u^T\rh}^*-v^T\rh \right)\rightarrow 0$, as $T$ goes to infinity.
\qed\end{proof}

\begin{corollary}\label{cor::trunerr1}
If ${v^{T}\rh}^*\ge0$ then  $\rho+2\cdot{u^{T}\rh}^*\ge \max_{\pi'\in A}\left\{\rev(\pi')\right\}$.
\end{corollary}
\begin{proof}
The proofs in Appendix~\ref{appendix::correctness} did not use the truncation of the process. We can therefore follow the same steps as in the proof of Lemma~\ref{lem::slope_of_rev},
\emph{Part II}: Put $\epsilon=2\cdot{u^{T_0}\rh}^*$. Then $v^*\rh\ge{v^{T_0}\rh}^*\ge0$, hence ${u^T\rh}^*\ge {v\rh}^*$, by Proposition~\ref{prop::trunerr}. Similarly to the implication following~(\ref{eq::epsilon_half0})-(\ref{eq::epsilon_half2}), we can deduce that $\rho+2\cdot {u^{T_0}\rh}^*\ge\max_{\pi'\in A}\left\{\rev(\pi')\right\}$.
\qed
\end{proof}

\noindent{\bf Proposition~\ref{prop::trunerr2}.}
\emph{
If $u$ and $\rho'$ are the outcome of the computation in Algorithm~\ref{alg}, lines~\ref{line_low_rho_assign}-\ref{line_mdp_solver_2}, then $\rho'+2\cdot (u+\epsilon') > \max_{\pi'\in A}\left\{\rev(\pi')\right\}$.}
\begin{proof}
If $low\le\epsilon/4$ then $\rho'$ is assigned the value $0$. In this case, as shown above, ${v^{T_0}_{1}}^*>0$. Assume that $low>\epsilon/4$.
In the proof of Proposition~\ref{prop::correctness} it was shown that the value returned by $mdp\_solver(\mdp_{low}^{T_0},\epsilon/8)$ is positive. Therefore, ${v^{T_0}_{low}}^*>-\epsilon/8$. Applying the proof of Lemma~\ref{lem::slope_of_rev} we deduce that ${v^{T_0}_{low-2\cdot\epsilon/8}}^*+\epsilon/8>\epsilon/8$, hence ${v^{T_0}_{\rho'}}^*>0$. Corollary~\ref{cor::trunerr1} thus applies to $\rho'$, and we obtain $\rho'+2\cdot{u^{T}\rh}^*\ge \max_{\pi'\in A}\left\{\rev(\pi')\right\}$. Observing that $u+\epsilon'>{u^{T_0}\rh}^*$ completes the proof.
\qed
\end{proof}

We complete the appendix with the proof of Corollary~\ref{cor::alg_4_threshold}:\\
\noindent{\bf Corollary~\ref{cor::alg_4_threshold}:}
\emph{Fix $\gamma$ and $\alpha$. If $u$ is the value returned by $mdp\_solver(\widehat{N_{\alpha}^T},\epsilon)$, and $u\le-\epsilon$, then \hon is optimal for $\alpha$. In other words, $\hat\alpha(\gamma)\ge\alpha$.
}
\begin{proof}
If $u\le-\epsilon$ then the value of $\widehat{N_{\alpha}^T}$ is smaller than $0$. 
If we denote by $\widehat{\mdp_{\alpha}}$ the same modification (of disabling honest mining) applied now to $\mdp_\alpha$, then the value of $\widehat{\mdp_{\alpha}}$ cannot be poisitive (similarly to Proposition~\ref{prop::trunerr}).
However, honest mining guarntees a value of $0$ in $\mdp_{\alpha}$, and we conclude that honest mining (weakly) dominates other strategies.
\qed
\end{proof}

\end{document}